\documentclass{eptcs}
\providecommand{\event}{FOCLASA 2012}

\usepackage{amsmath}
\usepackage{amsfonts}
\usepackage{amssymb}
\usepackage{amsthm}
\usepackage{color}
\usepackage{graphicx}
\usepackage{hyperref}
\usepackage{ifthen}
\usepackage{stmaryrd}
\usepackage{subfig}
\usepackage{tikz-reo}
\usepackage{xspace}

\newboolean{isTechReport}
\setboolean{isTechReport}{false}

\newcommand{\FIFO}[2][1]		{\textmd{\textsf{FIFO{\small #1}}}\ifthenelse{\equal{#2}{}}{}{$\langle #2 \rangle$}}
\newcommand{\LossyFIFO}[2][1]	{\textmd{\textsf{LossyFIFO{\small #1}}}\ifthenelse{\equal{#2}{}}{}{$\langle #2 \rangle$}}
\newcommand{\Alternator}[2][1]	{\textmd{\textsf{Alternator}}\ifthenelse{\equal{#2}{}}{}{$\langle #2 \rangle$}}
\newcommand{\SyncFIFOSync}[2][1]{\textmd{\textsf{SyncFIFO{\small #1}Sync}}\ifthenelse{\equal{#2}{}}{}{$\langle #2 \rangle$}}

\newcommand{\Fifo}[2][1]		{\textmd{\textsf{fifo{\small #1}}}\ifthenelse{\equal{#2}{}}{}{$\langle #2 \rangle$}}
\newcommand{\LossySync}[2][]	{\textmd{\textsf{lossysync}}\ifthenelse{\equal{#2}{}}{}{$\langle #2 \rangle$}}
\newcommand{\Sync}[2][]			{\textmd{\textsf{sync}}\ifthenelse{\equal{#2}{}}{}{$\langle #2 \rangle$}}
\newcommand{\SyncDrain}[2][]	{\textmd{\textsf{syncdrain}}\ifthenelse{\equal{#2}{}}{}{$\langle #2 \rangle$}}
\newcommand{\SyncSpout}[2][]	{\textmd{\textsf{syncspout}}\ifthenelse{\equal{#2}{}}{}{$\langle #2 \rangle$}}

\newcommand{\procfifo}[2][1]	{\mathit{Fifo\mbox{\small\textit{#1}}}\langle\mathit{#2}\rangle}
\newcommand{\proclossysync}[1]	{\mathit{LossySync}\langle\mathit{#1}\rangle}
\newcommand{\procsync}[1]		{\mathit{Sync}\langle\mathit{#1}\rangle}
\newcommand{\procsyncdrain}[1]	{\mathit{SyncDrain}\langle\mathit{#1}\rangle}

\newcommand{\procmerger}[1]		{\mathit{Merger}\langle\mathit{#1}\rangle}
\newcommand{\procreplicator}[1]	{\mathit{Replicator}\langle\mathit{#1}\rangle}
\newcommand{\procpumpstat}[1]	{\mathit{PumpingStation}\langle\mathit{#1}\rangle}
\newcommand{\procboundary}[1]	{\mathit{Boundary}\langle\mathit{#1}\rangle}

\newcommand{\reonode}			{n}
\newcommand{\reodata}			{d}

\newcommand{\actuniv}{\mathbb{A}\mathrm{ct}}
\newcommand{\A}			{a^{\silent}}
\newcommand{\B}			{b^{\silent}}
\newcommand{\C}			{c^{\silent}}
\newcommand{\ALPHA}{\alpha^{\delta}}
\newcommand{\BETA}{\beta^{\dead}}
\newcommand{\GAMMA}{\gamma^{\dead}}
\newcommand{\co}[1]{\overline{#1}}
\newcommand{\false}{\mathit{false}}
\newcommand{\mact}{\sqcup}
\newcommand{\mactalph}{{\ltimes}}
\newcommand{\mactin}{\sqsubseteq}
\newcommand{\mactindep}{\smallsmile}
\newcommand{\mactminus}{\setminus}
\newcommand{\mactpowset}{{\Downarrow}}
\newcommand{\mactquant}{\bigsqcup}
\newcommand{\mactsize}{\mathsf{sz}}
\newcommand{\mactuniv}{\mathbb{MA}\mathrm{ct}}
\newcommand{\silent}{\tau}
\newcommand{\silentact}{\mathtt{tau}}
\newcommand{\true}{\mathit{true}}

\newcommand{\dead}			{\delta}
\newcommand{\ch}			{+}
\newcommand{\seq}			{\cdot}
\newcommand{\parr}			{\mathop{\|}}
\newcommand{\lmerge}		{\mathop{\llfloor}}
\newcommand{\sync}			{\mid}
\newcommand{\restr}[1][V]	{\nabla_{#1}}
\newcommand{\block}[1][B]	{\partial_{#1}}
\newcommand{\rename}[1][R]	{\rho_{#1}}
\newcommand{\comm}[1][C]	{\Gamma_{#1}}
\newcommand{\hide}[1][I]	{\mathcal{T}_{#1}}

\newcommand{\commfun}[1][C]	{\mathcal{C}_{#1}}
\newcommand{\commdom}		{\mathsf{dom}}

\newcommand{\seqprocuniv}{\mathbb{S}\mathrm{eq}}

\newcommand{\acts}{\mathsf{Acts}}
\newcommand{\silentfree}{\silent\mbox{-}\mathsf{free}}

\newcommand{\proceq}{\simeq}
\newcommand{\xproceq}[1]{\stackrel{\mbox{\tiny #1}}{\proceq}}

\newcommand{\substfun}[1][w]		{\xi_{#1}}
\newcommand{\cosubstfun}[1][w]		{\co{\substfun[]}_{#1}}
\newcommand{\one}					{\mathtt{1}}
\newcommand{\two}					{\mathtt{2}}
\newcommand{\substenv}				{\Xi}
\newcommand{\substenvdom}			{\mathrm{dom}}
\newcommand{\substenvimg}			{\mathrm{img}}
\newcommand{\substenvcomm}			{\mathrm{comm}}

\newcommand{\hatisol}[1][\substenv]	{\widehat{\isol[]}_{#1}}
\newcommand{\isol}[1][\substenv]	{\mathsf{isol}_{#1}}
\newcommand{\coisol}[1][\substenv]	{\co{\isol[]}_{#1}}

\newcommand{\bisect}[1][\substenv]		{\mathsf{split}_{#1}}
\newcommand{\refbisect}[1][\substenv]	{\mathsf{SPLIT}_{#1}}

\newcommand{\qmark}[1][\substenv]	{\mathop{?}_{#1}}

\newcommand{\eqproof}[2][]{
	\begin{flushleft}
		\fcolorbox{white}{white}{
			$\begin{array}[t]{@{} c @{} l @{}}
				\eqeqbox{#1} & #2
			\end{array}$
		}
	\end{flushleft}
}

\newcommand{\eqprooff}[3]{
	\begin{flushleft}
		\fcolorbox{white}{white}{
			$#1\begin{array}[t]{@{} c @{} l @{}}
				\eqeqbox{#2} & #3
			\end{array}$
		}
	\end{flushleft}
}

\newcommand{\eqeqbox}[2][1.3cm]{\hbox to #1 {\hfil\ensuremath{#2}\hfil}}

\newcommand{\propref}[3][Prop.]{#1~\ref{prop:#2}\ifthenelse{\equal{#3}{}}{}{:\ref{prop:#2:#3}}}
\newcommand{\lemmaref}[3][Lem.]{#1~\ref{lemma:#2}\ifthenelse{\equal{#3}{}}{}{:\ref{lemma:#2:#3}}}

\newcommand{\xeq}[1]{\stackrel{\mbox{\tiny #1}}{=}}

\newcommand{\AND}{\mbox{ \textbf{\textup{and}} }}

\newcommand{\IMPLIES}{\mbox{ \textbf{\textup{implies}} }}
\newcommand{\OR}{\mbox{ \textbf{\textup{or}} }}

\newcommand{\binop}{\oplus}
\newcommand{\fun}{\dagger}

\newcommand{\set}[1]{\ensuremath{\{#1\}}}
\newcommand{\setbuild}[2]{\set{#1 \mid #2}}

\newcommand{\emptystring}{\epsilon}

\theoremstyle{plain}
\newtheorem{definition}{Definition}
\newtheorem{proposition}{Proposition}
\newtheorem{lemma}{Lemma}
\newtheorem{theorem}{Theorem}
\newtheorem{corollary}{Corollary}

\title{A Procedure for Splitting Processes \\ and its Application to Coordination \ifthenelse{\boolean{isTechReport}}{\\(Technical Report)}{}\footnotetext{%
	This research is partly funded by the EU project FP7-231620 HATS: Highly Adaptable and Trustworthy Software using Formal Models (\url{http://www.hats-project.eu/})
}}

\author{%
Sung-Shik T.Q. Jongmans
\institute{%
	Centrum Wiskunde \& Informatica
\\	Amsterdam, the Netherlands}
\email{jongmans@cwi.nl}
\and
Dave Clarke
\institute{%
	IBBT-DistriNet
\\	Department of Computer Science
\\	Katholieke Universiteit Leuven
\\	Leuven, Belgium}
\email{dave.clarke@cs.kuleuven.be}
\and
Jos\'{e} Proen\c{c}a
\institute{%
	IBBT-DistriNet
\\	Department of Computer Science
\\	Katholieke Universiteit Leuven
\\	Leuven, Belgium}
\email{jose.proenca@cs.kuleuven.be}
}

\def\titlerunning{A Procedure for Splitting Processes and its Application to Coordination}
\def\authorrunning{S.-S.T.Q. Jongmans, D. Clarke \& J. Proen\c{c}a}

\ifthenelse{\boolean{isTechReport}}{\newcommand{\publicationstatus}{}}{}

\begin{document}
\maketitle

\begin{abstract}
	We present a procedure for splitting processes in a process algebra with multi-actions (a subset of the specification language mCRL2). This splitting procedure cuts a process into two processes along a set of actions A: roughly, one of these processes contains no actions from A, while the other process contains only actions from A. We state and prove a theorem asserting that the parallel composition of these two processes equals the original process under appropriate synchronization.

	We apply our splitting procedure to the process algebraic semantics of the coordination language Reo: using this procedure and its related theorem, we formally establish the soundness of splitting Reo connectors along the boundaries of their (a)synchronous regions in implementations of Reo. Such splitting can significantly improve the performance of connectors.
\end{abstract}

\raggedbottom
\interfootnotelinepenalty=10000

%
%%
%%%
\section{Introduction}
\label{sect:intr}

\begin{figure}[t]
	\newcommand{\HEIGHT}	{57pt}
	\newcommand{\SCALE}		{1}
	
	\hfil
	\subfloat[{\FIFO[2]{}}]{\label{fig:conn:fifo2}
		\vbox to \HEIGHT {%
			\vfil
			\hbox {
				\scalebox{\SCALE}{ \begin{tikzpicture}[baseline, node distance=1.5cm]
	\node[reobnode,label=left:\small $a$]	(A) [] {};
	\node[reonode,label=right:\small $x$]	(X) [right of=A] {};
	\node[reobnode,label=right:\small $b$]	(B) [below of=X] {};
	
	\draw[fifo]	(A) to node {} (X);
	\draw[fifo]	(X) to node {} (B);
\end{tikzpicture}}
			}
			\vfil
		}
	}
	\hfil
	\subfloat[{\LossyFIFO[]{}}]{\label{fig:conn:lossyfifo}
		\vbox to \HEIGHT {%
			\vfil
			\hbox {
				\scalebox{\SCALE}{ \begin{tikzpicture}[baseline, node distance=1.5cm]
	\node[reobnode,label=left:\small $a$]	(A) [] {};
	\node[reonode,label=right:\small $x$]	(X) [right of=A] {};
	\node[reobnode,label=right:\small $b$]	(B) [below of=X] {};
	
	\draw[lossysync]	(A) to node {} (X);
	\draw[fifo]			(X) to node {} (B);
\end{tikzpicture}}
			}
			\vfil
		}
	}
	\hfil
	\subfloat[\Alternator{}]{\label{fig:conn:alternator}
		\vbox to \HEIGHT {%
			\vfil
			\hbox {
				\scalebox{\SCALE}{\begin{tikzpicture}[baseline, node distance=1.5cm]
	\node[reobnode,label=right:\small $c$]	(C) [] {};
	\node[reobnode,label=left:\small $a$]	(A) [left of=C] {};
	\node[reobnode,label=left:\small $b$]	(B) [below of=A] {};
	
	\draw[sync]			(A) to node {} (C);
	\draw[fifo]			(B) to node {} (C);
	\draw[syncdrain]	(A) to node {} (B);
\end{tikzpicture}}
			}
			\vfil
		}
	}
	\hfil
	\subfloat[{\SyncFIFOSync[]{}}]{\label{fig:conn:syncfifosync}
		\vbox to \HEIGHT {%
			\vfil
			\hbox {
				\scalebox{\SCALE}{\begin{tikzpicture}[baseline, node distance=1.5cm]
	\node[reobnode,label=right:\small $b$]	(B) [] {};
	\node[reobnode,label=left:\small $a$]	(A) [left of=B] {};
	\node[reonode,label=left:\small $x$]	(X) [below of=A] {};
	\node[reonode,label=right:\small $y$]	(Y) [right of=X] {};
	
	\draw[sync]	(A) to node {} (X);
	\draw[fifo]	(X) to node {} (Y);
	\draw[sync]	(Y) to node {} (B);
\end{tikzpicture}}
			}
			\vfil
		}
	}
	\hfil

	\caption{Some example connectors.}
	\label{fig:conn}
\end{figure}
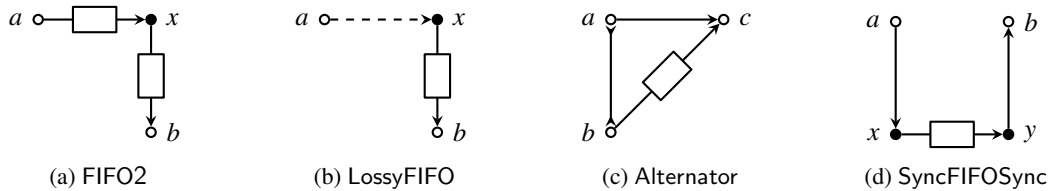

Over the past decades, coordination languages have emerged for the specification and implementation of interaction protocols among entities running concurrently (components, services, threads, etc.). This class of languages includes Reo~\cite{Arb04}, a graphical language for compositional construction of \emph{connectors}: communication mediums through which entities can interact with each other. Figure~\ref{fig:conn} shows some example connectors in their usual graphical syntax. Intuitively, connectors consist of one or more \emph{channels}, through which data items flow, and two or more \emph{nodes}, on which channel ends coincide. Through channel composition---the act of gluing channels together on nodes---engineers can construct complex connectors. Channels often used include the reliable synchronous channel, called \Sync{}, and the reliable asynchronous channel \Fifo[$n$]{}, which has a buffer of capacity $n$. Importantly, while nodes have a fixed semantics, Reo features an open-ended set of channels. This allows engineers to define their own channels with custom semantics.

To use connectors in real applications, one must derive executable code from graphical specifications of connectors (e.g., those in Figure~\ref{fig:conn}). Roughly two implementation approaches exist. In the \emph{distributed approach}, one implements the behavior of each of the $k$ constituents of a connector and runs these $k$ implementations concurrently as a distributed system; in the \emph{centralized approach}, one computes the behavior of a connector as a whole, implements this behavior, and runs this implementation sequentially as a centralized system. Neither of these two approaches unconditionally predominates the other: among other factors of influence, the hardware architecture on which to deploy the application plays an important role. For example, in the case of a service-oriented application, the distributed approach seems natural, because the services involved run on different machines and the network between them may play a role in their coordination. However, if coordination involves threads running on the same machine in some multi-threading application, the centralized approach appears more appropriate, as it avoids communication among the constituents of a connector at runtime: in this scenario, one dedicates one thread to running the connector.

One optimization technique applicable to both of these approaches involves the identification of the \emph{synchronous} and the \emph{asynchronous regions} of a connector. A synchronous region contains exactly those nodes and channels of a connector that synchronize collectively to decide on their individual behavior; an asynchronous region connects synchronous regions in an asynchronous way. For instance, the connector consisting of a \Sync{} channel, a \Fifo{} channel, and another \Sync{} channel (see Figure~\ref{fig:conn:syncfifosync}) has two synchronous regions, connected by an asynchronous region. Intuitively, two synchronous regions can run completely indepedently of each other;%
\footnote{%
	To see this, suppose that two synchronous regions cannot run completely independently of each other. In that case, there exist at least one constituent of the one region that synchronizes with at least one constituent of the other region. But then, these two constituents belong to the same synchronous region---a contradiction.
}
an asynchronous regions connecting them takes care of transporting data from one synchronous region to the other. In the distributed approach, this means that nodes and channels need to share information only with those nodes and channels in the same synchronous region---not with every node or channel in the connector. In the centralized approach, this means that one does not need to compute the behavior of a connector as a whole, but rather on a per-region basis.

Recent work shows that the optimization based on identifying regions can significantly improve performance~\cite{CP12,PCdVA12,Pro11}. However, while intuitively valid, a formal argument establishing the soundness of this optimization does not exist yet. In this paper, we present such a proof, based on the process algebraic semantics of Reo~\cite{KKdV??,KKdV10a,KKdV10b,KKdV10c}. In this semantics, one associates every connector with a process term describing its behavior. More concretely, we identify the following contributions:
\begin{itemize}
	\item
	We introduce a splitting procedure for a subset of the specification language mCRL2~\cite{GMR+08,GM10}---the basis of the existing process algebraic semantics of Reo---and prove its soundness.
	
	\item
	We formalize the notion of (a)synchronous regions in the process algebraic semantics of Reo.
	
	\item
	We apply this splitting procedure to the process algebraic semantics of Reo, thereby justifying the (a)synchronous regions optimization for Reo implementations. In particular, we discuss how we can implement and use the splitting procedure in the distributed approach, exploiting the local concurrency available on the computational nodes.
	
	\item
	We lay the foundations for the definition and analysis of new splitting operations for Reo.
\end{itemize}

This paper is organized as follows. In Section~\ref{sect:mcrl2}, we give an overview of the fragment of mCRL2 that we use. In Section~\ref{sect:reo}, we summarize the process algebraic semantics of Reo. In Section~\ref{sect:split}, we introduce our splitting procedure, and in Section~\ref{sect:appl}, we apply this procedure to connectors. We conclude this paper with future work in Section~\ref{sect:conc}. \ifthenelse{\boolean{isTechReport}}{}{See \cite{JCP12} for a version of this paper with an appendix including full proofs for all the intermediary lemmas.}

%
%%
%%%
\section{A Process Algebra with Multi-Actions}
\label{sect:mcrl2}

The process algebra considered in this work comprises the data-free and untimed fragment of mCRL2, a specification language based on ACP \cite{BK84} and the basis of the existing process algebraic semantics of Reo. Among other useful constructs, mCRL2 has one feature that makes it particularly well-suited as a semantic formalism for Reo, namely \emph{multi-actions}: collections of actions that occur at the same time. We postpone an explanation of how to use multi-actions for describing the behavior of connectors until Section~\ref{sect:reo}. In this section, we summarize (our subset of) mCRL2.

\begin{figure}[t]
	\newcommand{\HEIGHT}{55pt}
	\newcommand{\separator}{\enspace|\enspace}
	
	\hfil
	\subfloat[Syntax of multi-actions.]{\label{fig:syntax:mact}
		\fbox{
			\vbox to \HEIGHT {%
				\vfil
 				\hbox {
					$\begin{array}{@{} l @{\enspace} c @{\enspace} l @{}}
						a				& ::=	& \mbox{any element from } \actuniv
					\\	\A				& ::=	& a \separator \silent
					\\	\alpha , \beta	& ::=	& \A \separator \alpha \mact \beta
					\end{array}$
				}
				\vfil
			}
		}
	}
	\hfil
	\subfloat[Syntax of processes.]{\label{fig:syntax:proc}
		\fbox{
			\vbox to \HEIGHT {%
				\vfil
				\hbox {
					$\begin{array}{@{} l @{\enspace} r @{\enspace} l @{}}
						\ALPHA	& ::=	& \alpha \separator \dead 
					\\	p		& ::=	& \ALPHA \separator P \separator p \ch q \separator p \seq q
					\\			& |		& p \parr q \separator p \lmerge q \separator p \sync q 
					\\			& |		& \restr(p) \separator \block(p) \separator \rename(p) \separator \comm(p) \separator \hide(p)
					\end{array}$
				}
				\vfil
			}
		}
	}
	\hfil
	
	\caption{Syntax.}
	\label{fig:syntax}
\end{figure}

Figure~\ref{fig:syntax:mact} shows the syntax of multi-actions. Let $\actuniv$ denote the set of actions, ranged over by the symbols $a$, $b$, $c$, etc. The distinguished symbol $\silent$ denotes the empty multi-action, i.e., the multi-action consisting of no observable actions. Let the symbols $\A$, $\B$, $\C$, etc., range over the elements in the set $\actuniv \cup \set{\silent}$. The operator~$\mact$ (commutative and associative) combines multi-actions to form larger multi-actions; let $\mactuniv$ denote the set of all multi-actions, ranged over by $\alpha$, $\beta$, $\gamma$, etc. Processes, ranged over by $p$, $q$, $r$, etc., combine multi-actions using the operators shown in Figure~\ref{fig:syntax:proc}.
\begin{description}
	\item[Basic operators]
	The distinguished symbol---or nullary operator---$\dead$ denotes the deadlock process, i.e., the process performing no multi-actions. Let the symbols $\ALPHA$, $\BETA$, $\GAMMA$, etc., range over the processes in the set $\mactuniv \cup \set{\dead}$. The operators~$\ch$ and~$\seq$ combine processes alternatively and sequentially in the usual way.%
	\footnote{%
		We skip the basic operators for conditional composition and summation, because they have no meaning in the data-free fragment of mCRL2 considered. Similarly, we skip those operators that have no meaning in the untimed fragment of mCRL2.
	}
	Let $\seqprocuniv$ denote the set of sequential processes, which consist only of basic operators and multi-actions. Finally, let $P$, $Q$, $R$, etc., denote references that refer to process definitions of the form $P \mapsto p$, $Q \mapsto q$, $R \mapsto r$, etc. For technical convenience, we currently disallow mutual recursion: %impose the following restriction on the use of references:
if $P \mapsto p$, then only $P$ can occur as a reference in $p$.
	
	\item[Parallel operators]
	The operator~$\parr$ interleaves and synchronizes processes. The operator $\lmerge$ serves as an auxiliary operator in the axiomatization of~$\parr$: it makes the process on its left-hand side perform a multi-action, and afterwards, it combines the remaining process with the process on its right-hand side the same way $\parr$ does. The operator~$\sync$ synchronizes processes on the first multi-actions they perform, and it combines the remaining processes the same way $\parr$ does.
	
	\item[Additional operators]
	Four additional operators constrain the behavior of processes composed in parallel. The operator~$\restr[]$ restricts a process $p$ to the multi-actions in a set of nonempty multi-actions $V \subseteq \mactuniv \setminus \set{\silent}$ (modulo commutativity and associativity of $\mact$). The operator~$\block[]$ blocks those actions in a process $p$ that occur also in a set of actions $B \subseteq \actuniv$. The operator~$\rename[]$ renames the actions in a process $p$ according to a set of renaming rules $R \subseteq \actuniv \times \actuniv$. Finally, the operator~$\comm[]$ applies the communications in a set $C \subseteq \mactuniv \times \actuniv$ to a process $p$. We write communication rules as $\alpha \rightarrow a$ and require that $\silent$ does not occur in $\alpha$.
	
	\item[Abstraction operator]
	The operator~$\hide[]$ hides those actions in a process $p$ that occur also in a set of actions $I \subseteq \actuniv$. The act of hiding an action $a$, which means ``replacing $a$ by $\silent$,'' differs from the act of blocking $a$, which means ``replacing $a$ by $\dead$.''
\end{description}
\noindent We adopt the following usual operator precedence (in decreasing order): $\mact , \sync , \seq , \parr , \lmerge , \ch$. We write as few parentheses as possible, omitting them also in the case of associative or commutative operators. For example, we write $a \seq b \seq c \ch d \ch e$ instead of $(a \seq (b \seq c)) \ch (d \ch e)$.

See Section~\ref{sect:axioms} for an axiomatization of the operators discussed above.

%
%%
%%%
\section{Reo and its Process Algebraic Semantics}
\label{sect:reo}

Before we continue with our splitting procedure in Section~\ref{sect:split}, we briefly discuss Reo and its process algebraic semantics \cite{KKdV??,KKdV10a,KKdV10b,KKdV10c}; this helps in relating the abstract discussion in Section~\ref{sect:split} to a concrete case. Recall from Section~\ref{sect:intr} that connectors consist of channels and nodes. Below, we outline how these channels and nodes behave and how to describe such behavior as procesess.

\paragraph{Channels.}

\begin{figure}
	\centering
	\newcommand{\RAISE}{.5ex}
	\newcommand{\WIDTH}{351pt}
	\fbox{
		\begin{minipage}{.95\textwidth}
		\scalebox{.82}{
			\begin{tabular}{@{} c @{$\quad$} l @{$\quad$} l @{}}
				\emph{Graphical syntax}									& \emph{Textual syntax}	& \emph{Semantics}
			\\	\raisebox{\RAISE}{\begin{tikzpicture}[baseline, node distance=2cm]
	\node[invis,label=above right:$a$]		(I) [] {};
	\node[invis,label=above left:$b$]		(O) [right of=I] {};

	\draw[sync]	(I) to node {} (O);
\end{tikzpicture}}		& \Sync{a;b}			& \begin{minipage}[t]{\WIDTH}
																									\raggedright
																									Atomically accepts an item on its source end $a$ and dispenses it on its sink end $b$.
																								\end{minipage}
			\\	\raisebox{\RAISE}{\begin{tikzpicture}[baseline, node distance=2cm]
	\node[invis,label=above right:$a$]		(I) [] {};
	\node[invis,label=above left:$b$]		(O) [right of=I] {};

	\draw[lossysync]	(I) to node {} (O);
\end{tikzpicture}}	& \LossySync{a;b}		& \begin{minipage}[t]{\WIDTH}
																									\raggedright
																									Atomically accepts an item on its source end $a$ and, non-deterministically, either dispenses it on its sink end $b$ or loses it.
																								\end{minipage}
			\\	\raisebox{\RAISE}{\begin{tikzpicture}[baseline, node distance=2cm]
	\node[invis,label=above right:$a$]		(I) [] {};
	\node[invis,label=above left:$b$]		(O) [right of=I] {};

	\draw[syncdrain]	(I) to node {} (O);
\end{tikzpicture}}	& \SyncDrain{a,b;}		& \begin{minipage}[t]{\WIDTH}
																									\raggedright
																									Atomically accepts (and loses) items on both of its source ends $a$ and $b$.
																								\end{minipage}
			\\	\raisebox{\RAISE}{\begin{tikzpicture}[baseline, node distance=2cm]
	\node[invis,label=above right:$a$]		(I) [] {};
	\node[invis,label=above left:$b$]		(O) [right of=I] {};

	\draw[fifo]			(I) to node {} (O);
\end{tikzpicture}}		& \Fifo{a;b}			& \begin{minipage}[t]{\WIDTH}
																									\raggedright
																									Atomically accepts an item on its source end and stores it in its buffer, and atomically dispenses the item $\reodata$ on its sink end and clears its buffer.
																								\end{minipage}
			\end{tabular}
		}
		\end{minipage}
	}
	
	\caption{Syntax and semantics of common channels.}
	\label{fig:channels}
\end{figure}

Every channel has exactly two ends, each of which has one of two types: \emph{source ends} accept data, while \emph{sink ends} dispense data. Besides this assumption on their number of ends, Reo makes no assumptions about channels. This means, for example, that Reo allows channels with two source ends. Figure~\ref{fig:channels} shows the graphical syntax of four common channels, a textual syntax, and an informal description of their behavior. In the process algebraic semantics of Reo, one associates every channel end with an action. For source ends, such an action represents the acceptance of data; for sink ends, it represents the dispersal of data. By combining these actions in multi-actions, one can describe channels that atomically accept and dispense data on their ends. For example, the following \emph{recursive} process definitions describe the behavior of the channels in Figure~\ref{fig:channels}.
\begin{center}
	$\begin{array}{@{} l @{\;} c @{\;} l @{\quad \quad \quad} l @{\;} c @{\;} l @{}}
		\procsync{a;b}			& \mapsto	& a \mact b \seq \procsync{a;b}					& \procsyncdrain{a,b;}	& \mapsto	& a \mact b \seq \procsyncdrain{a,b;}
	\\	\proclossysync{a;b}		& \mapsto	& (a \mact b \ch a) \seq \proclossysync{a;b}	& \procfifo{a;b}		& \mapsto	& a \seq b \seq \procfifo{a;b}
	\end{array}$
\end{center}
The definition $\procsync{a;b}$ models synchronous flow through channel ends $a$ and $b$, represented by the multi-action $a \mact b$. The definition $\proclossysync{a;b}$ models a (nondeterministic) choice between flow through ends $a$ and $b$ and flow through only $a$, represented by the multi-action $a \mact b \ch a$. The definition $\procfifo{a;b}$ models flow through $a$ followed by flow through $b$. The recursion found in each of the four process definitions above indicates that the channels modeled by them repeat their behavior indefinitely.

In this paper, we adopt the context-\emph{in}sensitive process algebraic semantics of Reo, originally based on \emph{constraint automata}~\cite{BSAR06}. In context-insensitive semantic formalisms, one cannot directly describe channels and connectors whose behavior depends not only on their internal state but also on the presence or absence of I/O operations---their \emph{context}. In contrast, one can describe such channels and connectors in semantic formalisms that do support context-sensitivity. For instance, a context-sensitive version of \LossySync{} should lose a data item only in the absence of I/O operations on its sink end. A context-sensitive process algebraic semantics of Reo exists, originally based on \emph{connector coloring} with three colors~\cite{CCA07}. However, because this semantics depends on the data component of mCRL2, we do not consider it in this paper. We remark that we could \emph{encode} a context-sensitive process algebraic semantics along the lines of~\cite{JKA11}, which makes our splitting procedure applicable also to context-sensitive channels and connectors. For simplicity, however, we do not pursue that in this paper. See \cite{JA12} for an extensive overview of context-insensitive and context-sensitive semantic formalisms for Reo.

\paragraph{Nodes}

Entities communicating through a connector perform I/O operations---writes and takes---on its nodes. Reo features three kinds of nodes: \emph{source nodes} on which only source ends coincide, \emph{sink nodes} on which only sink ends coincide, and \emph{mixed nodes} on which both kinds of channel end coincide. Nodes have the following semantics.
\begin{itemize}
	\item 
	A source node $n$ has \emph{replicator semantics}. Once an entity attempts to write a data item $\reodata$ on $\reonode$, this node first suspends this operation. Subsequently, $\reonode$ notifies the channels whose source ends coincide on $\reonode$ that it offers $\reodata$. Once each of these channels has notified $\reonode$ that it accepts $\reodata$, $\reonode$ resolves the write: atomically, $\reonode$ dispenses $\reodata$ to each of its coincident source ends.
	
	\item 
	A sink node $\reonode$ has \emph{nondeterministic merger semantics}. Once an entity attempts to take a data item from $\reonode$, this node first suspends this operation. Subsequently, $\reonode$ notifies the channels whose sink ends coincide on $\reonode$ that it accepts a data item. Once at least one of these channels has notified $\reonode$ that it offers a data item, $\reonode$ resolves the take: atomically, $\reonode$ fetches this data item from the appropriate channel end and dispenses it to the entity attempting to take. If multiple sink ends offer a data item, $\reonode$ chooses one of them nondeterministically.
	
	\item
	A mixed node $\reonode$ has \emph{pumping station semantics}: a combination of the replicator semantics and merger semantics discussed above, where fetching and dispensing occurs atomically.
\end{itemize}

In the process algebraic semantics of Reo, one associates each of the $m$ source ends of a node with an action $src_{1 \leq i \leq m}$ and each of its $n$ sink ends with an action $snk_{1\leq i \leq n}$. Then, one can describe nodes by combining the processes for a binary replicator (one sink end to two source ends), a binary merger (two sink ends to one source end), a one-to-one pumping station, and a process for boundary nodes:
\begin{center}
	$\begin{array}{@{} l @{\;} c @{\;} l @{}}
		\procreplicator{snk ; src_1 , src_2}	& \mapsto	& snk \mact src_1 \mact src_2 \seq \procreplicator{snk ; src_1 , src_2}
	\\	\procmerger{snk_1 , snk_2 ; src}		& \mapsto	& (snk_1 \mact src \ch snk_2 \mact src) \seq \procmerger{snk_1 , snk_2 ; src}
	\\	\procpumpstat{snk ; src}				& \mapsto	& snk \mact src \seq \procpumpstat{snk ; src}
	\\	\procboundary{bnd}						& \mapsto	& bnd \seq \procboundary{bnd}
	\end{array}$
\end{center}

\paragraph{Connectors.}

To get the behavior of a connector as a process, one composes the processes of the constituents of that connector in parallel and synchronizes their actions appropriately. Below, we give the processes of the connectors in Figures~\ref{fig:conn:fifo2} and~\ref{fig:conn:alternator}. See~\cite{KKdV??,KKdV10a,KKdV10b,KKdV10c} for more examples.

\begin{center}
	$\begin{array}{@{} l @{\;} c @{\;} l @{}}
		\mathit{Fig\ref{fig:conn:fifo2}}		& =	& \block[\set{a_1,\co{a}_1,x_1,\co{x}_1,x_2,\co{x}_2,b_1,\co{b}_1}](\comm[\set{a_1 \mact \co{a}_1 \rightarrow a , x_1 \mact \co{x}_1 \rightarrow x , x_2 \mact \co{x}_2 \rightarrow x , b_1 \mact \co{b}_1 \rightarrow b}](
	\\											&	& \quad \procboundary{\co{a}_1} \parr \procfifo{a_1;x_1} \parr \procpumpstat{\co{x}_1;\co{x}_2} \parr \procfifo{x_2;b_1} \parr \procboundary{\co{b}_1}))
	\\	\vspace{-2ex}
	\\	\mathit{Fig\ref{fig:conn:alternator}}	& =	& \block[\setbuild{\ast_{bnd},\co{\ast}_{bnd},\ast_i,\co{\ast}_i}{\ast \in \set{a,b,c} \wedge i \in \set{1,2}}](\comm[\setbuild{\ast_{bnd} \mact \co{\ast}_{bnd} \mact \ast_i \mact \co{\ast}_i \rightarrow \ast}{\ast \in \set{a,b,c} \wedge i \in \set{1,2}}](
	\\											&	& \quad \procboundary{a_{bnd}} \parr \procreplicator{\co{a}_{bnd} ; \co{a}_1,\co{a}_2} \parr \procboundary{b_{bnd}} \parr \procreplicator{\co{b}_{bnd} ; \co{b}_1,\co{b}_2} \parr {}
	\\											&	& \quad \procsyncdrain{a_2;b_2} \parr \procsync{a_1;c_1} \parr \procfifo{b_1;c_2} \parr \procmerger{\co{c}_1,\co{c}_2 ; \co{c}_{bnd}} \parr \procboundary{c_{bnd}}))
	\end{array}$
\end{center}

%
%%
%%%
\section{Splitting Processes}
\label{sect:split}

Recall from Section~\ref{sect:intr} that we aim at establishing the validity of optimizing implementations of Reo through the identification of (a)synchronous regions. Essentially, we want to show that \emph{splitting} connectors along the boundaries of their (a)synchronous regions (and running the resulting subconnectors concurrently) does not give rise to inadmissible behavior. In this section, we lay the foundation for this kind of splitting in terms of a splitting procedure for processes. Later, in Section~\ref{sect:appl}, we apply this procedure to the process algebraic semantics of Reo, thereby justifying the splitting of connectors. Here, we start by explaining the intuition behind our splitting procedure; formal definitions appear in Section~\ref{sect:split:defs}, followed by theorems and proofs in Section~\ref{sect:split:thm}. We note that our notion of ``splitting'' differs from ``decomposition'' in the spirit of \cite{MM93}: in our context, primality or uniqueness do not matter.

Let $\acts(p)$ denote the set of actions occurring in a process $p$. We introduce the function $\bisect[]$, which splits a process $p$ along a set of actions $A$ into two processes: one of these processes contains no actions in $\acts(p) \setminus A$, while the other process contains no actions in $A$. We call the former process the \emph{$A$-isolation of $p$} and the latter process the \emph{$A$-coisolation of $p$}. We aim at constructing $p$'s isolation and its coisolation such that their parallel composition equals $p$ under appropriate synchronization. Informally, to construct $p$'s $A$-isolation, replace every action in $p$ as follows:
\begin{itemize}
	\item 
	If $a \in A$, replace $a$ with the multi-action $a \mact \substfun[](a)$, where $\substfun[](a)$ denotes a fresh action with respect to $\acts(p)$. Intuitively, $\substfun[](a)$ represents the act of ``disseminating that this process performs $a$.''
	
	\item 
	If $b \notin A$, replace $b$ with the action $\cosubstfun[](b)$, where $\cosubstfun[](b)$ denotes a fresh action with respect to $\acts(p)$. Intuitively, $\cosubstfun[](b)$ represents the act of ``discovering that another process performs $b$.''
\end{itemize}
Symmetrically, to construct the $A$-coisolation of a process $p$, replace in $p$ every $b \in A$ with $\cosubstfun[](b)$ and every $b \notin A$ with $b \mact \substfun[](b)$. Note that because the foregoing affects only multi-actions, $p$'s isolation and its coisolation have the same structure as $p$. In other words: the process $p$, its isolation, and its coisolation have the same transition system modulo transition labels.

To illustrate isolation and coisolation, consider the process $q = a \seq b$ as a running example. This process has $q_1 = a \mact \substfun[](a) \seq \cosubstfun[](b)$ as its $\set{a}$-isolation and $q_2 = \cosubstfun[](a) \seq b \mact \substfun[](b)$ as its $\set{a}$-coisolation. 
% \begin{center}
% 	$q_1 \parr q_2 \; \proceq \; \begin{array}[t]{@{} l @{}}
% 		a \mact \substfun[](a) \seq (\cosubstfun[](b) \seq \cosubstfun[](a) \seq b \mact \substfun[](b) \ch {}
% 	\\	\quad \cosubstfun[](a) \seq (b \mact \substfun[](b) \seq \cosubstfun[](b) \ch b \mact \cosubstfun[](b) \seq \substfun[](b) \ch b \mact \substfun[](b) \mact \cosubstfun[](b)) \ch \cosubstfun[](b) \mact \cosubstfun[](a) \seq b \mact \substfun[](b)) \ch {}
% 	\\	\cosubstfun[](a) \seq (b \mact \substfun[](b) \seq a \mact \substfun[](a) \seq \cosubstfun[](b) \ch {} 
% 	\\	\quad a \mact \substfun[](a) \seq (\cosubstfun[](b) \seq b \mact \substfun[](b) \ch b \mact \substfun[](b) \seq \cosubstfun[](b) \ch \cosubstfun[](b) \mact b \mact \substfun[](b)) \ch b \mact \substfun[](b) \mact a \mact \substfun[](a) \seq \cosubstfun[](b)) \ch {}
% 	\\	a \mact \substfun[](a) \mact \cosubstfun[](a) \seq (\cosubstfun[](b) \seq b \mact \substfun[](b) \ch b \mact \substfun[](b) \seq \cosubstfun[](b) \ch \cosubstfun[](b) \mact b \mact \substfun[](b))
% 	\end{array}$
% \end{center}
However, the parallel composition of $q_1$ and $q_2$ is not equal to $q$ yet: to ensure that a process equals the parallel composition of its isolation and its coisolation, these latter two processes should synchronize on $\substfun[](a)$ and $\cosubstfun[](a)$ for each $a$. To this end, we apply the communication operator $\comm[]$ to such compositions. In our running example, this yields the process $\comm(q_1 \parr q_2)$ with $C = \set{\substfun[](a) \mact \cosubstfun[](a) \rightarrow \silentact \: , \: \substfun[](b) \mact \cosubstfun[](b) \rightarrow \silentact}$. The special action $\silentact$ serves as a placeholder action for $\silent$, and we can hide it immediately using the abstraction operator~$\hide[]$;%
\footnote{%
	We use this construction, because mCRL2 does not permit communications to map directly to $\silent$.
} 
henceforth, without loss of generality, we assume $\silentact \notin \acts(p)$ for each $p$. In our running example, this yields the process $\hide(\comm(q_1 \parr q_2))$ with $I = \set{\silentact}$ and $C$ as before.
% \begin{center}
% 	$\hide(\comm(q_1 \parr q_2)) \; \proceq \; \begin{array}[t]{@{} l @{}}
% 		a \mact \substfun[](a) \seq (\cosubstfun[](b) \seq \cosubstfun[](a) \seq b \mact \substfun[](b) \ch {}
% 	\\	\quad \cosubstfun[](a) \seq (b \mact \substfun[](b) \seq \cosubstfun[](b) \ch b \mact \cosubstfun[](b) \seq \substfun[](b) \ch b) \ch \cosubstfun[](b) \mact \cosubstfun[](a) \seq b \mact \substfun[](b)) \ch {}
% 	\\	\cosubstfun[](a) \seq (b \mact \substfun[](b) \seq a \mact \substfun[](a) \seq \cosubstfun[](b) \ch {} 
% 	\\	\quad a \mact \substfun[](a) \seq (\cosubstfun[](b) \seq b \mact \substfun[](b) \ch b \mact \substfun[](b) \seq \cosubstfun[](b) \ch b) \ch b \mact \substfun[](b) \mact a \mact \substfun[](a) \seq \cosubstfun[](b)) \ch {}
% 	\\	a \seq (\cosubstfun[](b) \seq b \mact \substfun[](b) \ch b \mact \substfun[](b) \seq \cosubstfun[](b) \ch b)
% 	\end{array}$
% \end{center}
But also this process is not equal to $q$ yet: only synchronization and abstraction do not suffice---we must also block those actions whose performance in isolation ``makes no sense.'' For instance, we consider every unpaired occurrence of $\cosubstfun[](a)$ in a multi-action $\alpha$ nonsensical: intuitively, performing $\cosubstfun[](a)$ suggests that some process discovers that another process performs $a$, even though this does not happen (otherwise, also $\substfun[](a)$ would occur in $\alpha$). By symmetry, we consider also every unpaired occurrence of $\substfun[](a)$ nonsensical. To block unpaired occurrences of $\substfun[](a)$ and $\cosubstfun[](a)$, we apply the blocking operator $\block[]$. In our running example, this yields the process $\block(\hide(\comm(q_1 \parr q_2)))$ with $B = \set{\substfun[](a) , \cosubstfun[](a) , \substfun[](b) , \cosubstfun[](b)}$ and $I$ and $C$ as before. This process equals $q$. 
% We refer the interested reader to Section~\ref{sect:derivs} for a complete derivation.

%
%%
\subsection{Formal Definitions}
\label{sect:split:defs}

\begin{figure}[t]
	\newcommand{\HEIGHT}{56pt}
	\newcommand{\separator}{\enspace|\enspace}
	
	\hfil
	\begin{minipage}{266pt}
		\fbox{
			\vbox to \HEIGHT {%
				\vfil
				\hbox{
					$\begin{array}{@{} l @{\enspace} c @{\enspace} l @{}}
						\substenvdom(\substenv)		& = & dom(\substfun) \cap dom(\cosubstfun)
					\\	\substenvimg(\substenv)		& = & img(\substfun) \cup img(\cosubstfun)
					\\	\substenvcomm(\substenv)	& = & \setbuild{\substfun(a) \mact \cosubstfun(a) \rightarrow \silentact}{ (a , w) \in \substenvdom(\substenv)}
					\end{array}$
				}
				\vfil
			}
		}
		\caption{Auxiliary functions for substitution environments.}
		\label{fig:auxfuns:substenv}
	\end{minipage}
	\hfil
	\begin{minipage}{159pt}
		\fbox{
			\vbox to \HEIGHT {%
				\vfil
				\hbox {
					\begin{tabular}{@{} r @{$\quad$} l @{}}
						\textup{Q1}	& $\qmark(\silent) \proceq \silent$
					\\	\textup{Q2}	& $\qmark(\dead) \proceq \dead$
					\\	\textup{Q3}	& $\qmark(p \ch q) \proceq \qmark(p) \ch \qmark(q)$
					\\	\textup{Q4}	& $\qmark(p \seq q) \proceq \qmark(p) \seq \qmark(q)$
					\end{tabular}
				}
				\vfil
			}
		}
		
		\caption{Axioms for $\qmark[]$.}
		\label{fig:axioms:qmark}
	\end{minipage}
	\hfil
\end{figure}

We proceed with formal definitions of the splitting procedure outlined above. We start with a formal account of the fresh auxiliary actions of the form $\substfun[](a)$ and $\cosubstfun[](a)$. As suggested by this notation, $\substfun[]$ and $\cosubstfun[]$ denote functions that take an action $a$ as their input and produce another action as their output. We collect such pairs of functions in \emph{substitution environments} as follows. Let $\set{\one , \two}^*$ denote the set of finite strings over $\set{\one , \two}$, ranged over by $w$, $v$, $u$, etc.
\begin{definition}
	\label{def:substenv}
	A substitution environment, typically denoted by $\substenv$, is a quintuple $(\dot P \mapsto \dot p , \mathbb{A} , \silentact , \substfun[] , \cosubstfun[])$ consisting of a process definition $\dot P \mapsto \dot p$, a set $\mathbb{A} \subseteq \actuniv$, an action $\silentact \in \actuniv \setminus \mathbb{A}$ and injective functions $\substfun[] , \cosubstfun[] : \mathbb{A} \times \set{\one , \two}^* \rightarrowtail \actuniv \setminus (\mathbb{A} \cup \silentact)$ such that $img(\substfun[]) \cap img(\cosubstfun[]) = \emptyset$.%
\end{definition}
\noindent Henceforth, we write $\substfun(a)$ and $\cosubstfun(a)$ instead of $\substfun[](a , w)$ and $\cosubstfun[](a , w)$. Note that we dropped the $w$ subscripts in our running example above: as we did not need this extra string of information, we omitted it for simplicity. In the general case, however, this information plays a key role, as explained shortly. The process definition in a substitution environment represents the main process to be split.

Figure~\ref{fig:auxfuns:substenv} shows auxiliary functions for substitution environment. The functions ``$\substenvdom$'' and ``$\substenvimg$'' map substitution environments to their domain and image. The function ``$\substenvcomm$'' maps substitution environments to communications derivable from them.

\begin{figure}[t]
	\centering 
	\fbox{
		\begin{minipage}{.75\linewidth}
		$\begin{array}{@{} l @{\;} c @{\;} l @{\enspace} l @{}}
			\isol(a, A , w)					& =	& a \mact \substfun(a)									& \mbox{if } a \in A
		\\	\isol(b, A , w)					& =	& \cosubstfun(b)										& \mbox{if } b \notin A
		\end{array}$
		\hfill
		$\begin{array}{@{} l @{\;} c @{\;} l @{\enspace} l @{}}
			\coisol(a, A , w)				& =	& \cosubstfun(a)										& \mbox{if } a \in A
		\\	\coisol(b, A , w)				& =	& b \mact \substfun(b)									& \mbox{if } b \notin A
		\end{array}$
		
		\centering
		$\begin{array}{@{} l @{\;} c @{\;} l @{\enspace} l @{}}
		\\	\hatisol(\vartheta , A , w)		& =	& \vartheta												& \mbox{for } \vartheta \in \set{\silent , \dead}
		\\	\hatisol(p \binop q , A , w)	& =	& \hatisol(p , A , w) \binop \hatisol(q , A , w)		& \mbox{for } \binop \in \set{\seq , \mact}
		\\	\hatisol(p \ch q , A , w)		& =	& \hatisol(p , A , w\one) \ch \hatisol(q , A , w\two)
		\end{array}$
		\end{minipage}
	}
	
	\caption{The functions $\isol[]$ and $\coisol[]$. Let $p \in \seqprocuniv$ and $\hatisol[] \in \set{\isol[] , \coisol[]}$.}
	\label{fig:isol,coisol}
\end{figure}

To formalize the notions of $A$-isolation and $A$-coisolation, we introduce the functions $\isol[]$ and $\coisol[]$. Figure~\ref{fig:isol,coisol} shows their definitions. The functions $\isol[]$ and $\coisol[]$ take for arguments a sequential process, a set of actions $A \subseteq \actuniv$, a string $w \in \set{\one , \two}^*$, and a substitution environment (as a subscript for notational convenience). For most processes $p$, $\isol(p , A , w)$ and $\coisol(p , A , w)$ invoke themselves recursively on $p$'s immediate subprocesses, the same set $A$, and the same string $w$. One exception exists: processes of the form $p \ch q$. For such processes, $\isol[]$ and $\coisol[]$ invoke themselves recursively on $w\one$ and $w\two$ instead of~$w$. This ensures that in their parallel composition, the process $\isol(p \ch q , A , w)$ can ``track'' which choice the process $\coisol(p \ch q , A , w)$ makes and vice versa. 

To clarify this, let us illustrate what would happen if $\isol(p \ch q , A , w)$ and $\coisol(p \ch q , A , w)$ invoke themselves recursively without changing $w$. In that case, $w$ has no influence on the behavior of $\isol[]$ and $\coisol[]$, and we can omit it from our definitions. Now, suppose that we want to compose the $\set{a}$-isolation and $\set{a}$-coisolation of the process $r = a \seq b \ch a \seq c$ in parallel. We have:
\begin{center}
	$\begin{array}{@{} l @{\;} c @{\;} r @{\,} r @{\,} r @{\,} l @{}}
		\isol(r , \set{a})		& =	&  														& a \mact \substfun[](a) \seq {}	& \cosubstfun[](b)			& {} \ch a \mact \substfun[](a) \seq \cosubstfun[](c)
	\\	\coisol(r , \set{a})	& =	& \cosubstfun[](a) \seq b \mact \substfun[](b) \ch {}	& \cosubstfun[](a) \seq {}			& c \mact \substfun[](c)
	\end{array}$
\end{center}
Thus, the process $\isol(r , \set{a})$ can erroneously synchronize its left-most multi-action $a \mact \substfun[](a)$ with the right-most multi-action $\cosubstfun[](a)$ of the process $\coisol(r , \set{a})$. By changing $w$ in the recursive invocations of $\isol(p \ch q , A , w)$ and $\coisol(p \ch q , A , w)$, this problem does not arise: it ensures that $a \mact \substfun[w\one](a)$ (on the left) can synchronize only with $\cosubstfun[w\one](a)$ (also on the left)---not with $\cosubstfun[w\two](a)$ (on the right). Note that this depends on the injectivity of $\substfun[]$ and $\cosubstfun[]$ (see Definition~\ref{def:substenv}).

The definition of the function $\bisect[]$ follows straightforwardly now that we have the functions $\isol[]$ and $\coisol[]$. We also introduce an auxiliary operator, denoted by $\qmark[]$, which encapsulates the communication, hiding, and blocking necessary to get equality of processes. Figure~\ref{fig:axioms:qmark} shows axioms for this operator.%
\footnote{%
	The axiom Q1 follows from the axioms C1, H1, and B1 in Figure~\ref{fig:axioms} in Section~\ref{sect:axioms}; Q2 follows from C2, H5, and B5; Q3 follows from C3, H6, and B6; Q4 follows from C4, H7, and B7.
}

\begin{definition}
	\label{def:qmark}
	$\qmark(p) = \block[\substenvimg(\substenv)](\hide[\set{\silentact}](\comm[\substenvcomm(\substenv)](p)))$
\end{definition}

\begin{definition}
	\label{def:bisect}
	For all $\substenv = (\dot P \mapsto \dot p , \mathbb{A} , \silentact , \substfun[] , \cosubstfun[])$ such that $\acts(p) \subseteq \mathbb{A}$,
	\begin{center}
		$\bisect(p , A , w) = \left\{ \begin{array}{@{} l @{\enspace} l @{}}
			\qmark(\isol(p , A , w) \parr {} \coisol(p , A , w))	& \mbox{if } p \in \seqprocuniv
		\\	\bisect(p_1 , A , w) \binop \bisect(p_2 , A , w)		& \mbox{if } p \notin \seqprocuniv \mbox{ and } p = p_1 \binop p_2 \mbox{ and } \binop \in \set{\seq , \ch , \parr , \lmerge , \sync}
		\\	\fun(\bisect(p_1 , A , w))								& \mbox{if } p = \fun(p_1) \mbox{ and } \fun \in \set{\restr , \block , \rename , \comm , \hide}
		\\	\refbisect(\dot P , A , w)								& \mbox{if } p = \dot P
		\end{array} \right.$
	\end{center}
	where $\refbisect(\dot P , A , w)$ denotes a reference to the process $\bisect(\dot p , A , w)$.
\end{definition}

\subsection{Theorems}
\label{sect:split:thm}

Suppose an execution environment $\substenv = (\dot P \mapsto \dot p , \mathbb{A} , \silentact , \substfun[] , \cosubstfun[])$. We prove that splitting $\dot p$ as described above yields a process equal to $\dot p$. We proceed in three steps. First, we prove our result for multi-actions. Then, we extend this result to sequential processes. Finally, we establish it for general processes. In each of these theorems we restrict our attention to \emph{syntactically $\silent$-free} specifications, because we work under strong bisimulation. Under equivalences weaker than strong bisimulation, we can relax this $\silent$-freeness. 

The axioms occasionally referred to in the remainder of this section appear in Figure~\ref{fig:axioms}, Section~\ref{sect:axioms}.

\subsubsection{A theorem for multi-actions}
\label{sect:split:thm:mact}

We start with a theorem for multi-actions, which states that splitting a syntactically $\silent$-free multi-action equals that multi-action. Let $\silentfree(\alpha)$ denote that $\silent$ does not occur in $\alpha$ (see \ifthenelse{\boolean{isTechReport}}{Section~\ref{sect:proofs}}{\cite{JCP12}} for a formal definition).
\begin{theorem}
	\label{theorem:bisect(mact):id}
	For all $\substenv = (\dot P \mapsto \dot p , \mathbb{A} , \silentact , \substfun[] , \cosubstfun[])$ such that $\acts(\alpha) \cup A \subseteq \mathbb{A}$,
	\begin{center}
		$\silentfree(\alpha) \IMPLIES \bisect(\alpha , A , w) \proceq \alpha$
	\end{center}
\end{theorem}

To prove this theorem, we need some auxiliary lemmas. We formulate these lemmas below; detailed proofs, as well as additional propositions on which these proofs rely, appear in \ifthenelse{\boolean{isTechReport}}{Section~\ref{sect:proofs}}{\cite{JCP12}}. The first lemma states that the parallel composition of the isolation and the coisolation of a process equals their synchronous composition (after applying communication, hiding, and blocking).
\begin{lemma}
	\label{lemma:qmark(isol,coisol):ms}
	For all $\substenv = (\dot P \mapsto \dot p , \mathbb{A} , \silentact , \substfun[] , \cosubstfun[])$ such that $\acts(p) \cup A \subseteq \mathbb{A}$,
	\begin{center}
		\begin{tabular}[t]{@{}l@{}}
			$\big[\silentfree(p) \AND p \in \seqprocuniv \big] \IMPLIES$
		\\	$\quad \qmark(\isol(p , A , w) \parr \coisol(p , A , w)) \proceq \qmark(\isol(p , A , w) \sync \coisol(p , A , w))$
		\end{tabular}
	\end{center}
\end{lemma}
\begin{proof}
	[Proof (sketch)]
	By the axioms M and A6, we must show that $\qmark(\isol(p , A , w) \lmerge \qmark(\coisol(p , A , w)$ and $\qmark(\coisol(p , A , w) \lmerge \qmark(\isol(p , A , w)$ equal $\dead$. Both of these processes start with a multi-action $\ALPHA$. By the definition of $\isol[]$ and $\coisol[]$ (and, in particular, the injectivity and image-disjointness of $\substfun[]$ and $\cosubstfun[]$), if $\ALPHA \neq \dead$, it must contain an action $\substfun(a)$ without $\cosubstfun(a)$ (or vice versa). But then, the blocking operator in $\qmark[]$ (combined with SMA) will equate $\ALPHA$ to $\dead$. This suffices to show that $\qmark(\isol(p , A , w) \lmerge \qmark(\coisol(p , A , w)$ and $\qmark(\coisol(p , A , w) \lmerge \qmark(\isol(p , A , w)$ equal $\dead$ by A7 (because these processes start with $\ALPHA$). See \ifthenelse{\boolean{isTechReport}}{Section~\ref{sect:proofs:lemma:qmark(isol,coisol):ms}}{\cite{JCP12}} for a detailed proof.
\end{proof}
\noindent Note that Lemma~\ref{lemma:qmark(isol,coisol):ms} involves sequential processes rather than only multi-actions. This enables us to use this lemma also in our proof of Theorem~\ref{theorem:bisect(seqproc):id}, below.

The following lemma consists of two parts. The first part states that one can rewrite every multi-action composed of the isolation and the coisolation of a multi-action $\alpha$ into a representation with the following characteristics: (i) for every communication $\beta_i \rightarrow \silentact$ induced by the substitution environment involved, $\beta_i$ occurs zero or more times; (ii) the \emph{remainder} $\breve \alpha$ does not contain any fragment of any $\beta_i$ and vice versa, denoted as $\breve \alpha \mactindep \beta_i$. (See \ifthenelse{\boolean{isTechReport}}{Section~\ref{sect:proofs}}{\cite{JCP12}} for a formal definition of the latter relation.) The second part of the following lemma states the additivity property $\comm(\alpha \mact \alpha_2) \proceq \comm(\alpha_1) \mact \comm(\alpha_2)$ when $\alpha_1$ and $\alpha_2$ each have such a representation. Let $\mactquant_n \beta$ denote the sequence $\beta \mact \cdots \mact \beta$ of length $n$.
\begin{lemma}
	\label{lemma:comm(mact):addit}
	$ $
 	\begin{enumerate}
 		\item 
 		For all $\substenv = (\dot P \mapsto \dot p , \mathbb{A} , \silentact , \substfun[] , \cosubstfun[])$ such that $\acts(\alpha) \cup A \subseteq \mathbb{A}$ and $\commdom(\substenvcomm(\substenv)) = \set{\beta_1 , \ldots , \beta_k}$,
 		\begin{center}
			$\isol(\alpha , A , w) \mact \coisol(\alpha , A , w) \proceq \mactquant_{n_1} \beta_1 \mact \cdots \mact {} \mactquant_{n_k} \beta_k \mact \breve \alpha \AND  \breve \alpha \mactindep \beta_i$
		\end{center}
		
		\item 
		For all $C = \set{\beta_1 \rightarrow b_1 , \ldots , \beta_k \rightarrow b_k}$,
		 \begin{center}
			$\left[ \begin{array}{@{} c @{} l @{\;}}
						& \alpha_1 = \mactquant_{n_1} \beta_1 \mact \cdots \mact {} \mactquant_{n_k} \beta_k ~\mkern1.5mu \mact \breve \alpha_1 \AND \breve \alpha_1 \mactindep \beta_i
			\\	\AND	& \alpha_2 = \mactquant_{m_1} \beta_1 \mact \cdots \mact {} \mactquant_{m_k} \beta_k \mact \breve \alpha_2 \AND \breve \alpha_2 \mactindep \beta_i
			\end{array} \right] \IMPLIES \left[ \begin{array}{@{} l @{}}
				\comm(\alpha_1 \mact \alpha_2) \proceq {}
			\\	\quad \comm(\alpha_1) \mact \comm(\alpha_2)
			\end{array} \right]$
		\end{center}
 	\end{enumerate}
\end{lemma}
\begin{proof}
	[Proof (sketch)]
	$ $
	\begin{enumerate}
		\item
		If $\alpha = a$, there exists a $\beta_\ell = \substfun(a) \mact \cosubstfun(a)$ for some $\ell$. By the definition of $\isol[]$ and $\coisol[]$, we have that $\isol(\alpha , A , w) \mact \coisol(\alpha , A , w) = \beta_\ell \mact a$. Identifying $\breve \alpha$ with $a$, we must show that $a$ does not occur in any $\beta_i$. This follows from the fact that $\substfun[]$ and $\cosubstfun[]$ have disjoint domains and images by their definition. The general case follows by structural induction.
		
		\item 
		Because $\breve \alpha_1$ and $\breve \alpha_2$ do not contain any fragment of any $\beta_i$ (i.e., $\breve \alpha_1 \mactindep \beta_i$ and $\breve \alpha_2 \mactindep \beta_i$), combining them in the same multi-action does not make the communication operator applicable: there exists no communication in $\alpha_1 \mact \alpha_2$ that did not exist already in $\alpha_1$ or in $\alpha_2$.
	\end{enumerate}
	See \ifthenelse{\boolean{isTechReport}}{Section~\ref{sect:proofs:lemma:comm(mact):addit}}{\cite{JCP12}} for a detailed proof. 
\end{proof}

\noindent The following corollary follows from the previous lemma: it asserts the additivity property $\comm(\alpha_1 \mact \alpha_2) \proceq \comm(\alpha_1) \mact \comm(\alpha_2)$ for $\alpha_1 = \isol(\alpha , A , w) \mact \coisol(\alpha , A , w)$ and $\alpha_2 = \isol(\beta , A , w) \mact \coisol(\beta , A , w)$.
\begin{corollary}
	\label{corol:comm(mact):addit}
	For all $\substenv = (\dot P \mapsto \dot p , \mathbb{A} , \silentact , \substfun[] , \cosubstfun[])$ such that $\acts(\alpha) \cup A \subseteq \mathbb{A}$,
	\begin{center}
		\begin{tabular}[t]{@{}l@{}}
			$\left[ \begin{array}{@{}l@{}}
				\comm[\substenvcomm(\substenv)](\isol(\alpha , A , w) \mact \coisol(\alpha , A , w) \mact \isol(\beta , A , w) \mact \coisol(\beta , A , w)) \proceq {}
			\\	\quad \comm[\substenvcomm(\substenv)](\isol(\alpha , A , w) \mact \coisol(\alpha , A , w)) \mact \comm[\substenvcomm(\substenv)](\isol(\beta , A , w) \mact \coisol(\beta , A , w))
			\end{array} \right]$
		\end{tabular}
	\end{center}
\end{corollary}
\noindent Finally, Figure~\ref{fig:proof:theorem:bisect(mact):id} shows a proof of Theorem~\ref{theorem:bisect(mact):id}.

\begin{figure}[t]
	\centering
	\small
	\fbox{
		\begin{minipage}{.95\textwidth}
			Suppose $\silentfree(\alpha)$ (Prem). We proceed by structural induction on $\alpha$.
			\begin{description}
				\item[Base:]
				$\alpha = \A$. If $\A = \silent$, we get a contradition with Prem. If $\A = a$:
				\eqprooff																{ \qmark(\bisect(\alpha , A , w))}
					{\xeq{$\alpha$,$\A$,$\bisect[]$}}									{ \qmark(\isol(a , A , w) \parr \coisol(a , A , w))
				\\	\xproceq{Prem$\rightarrow$Lem.~\ref{lemma:qmark(isol,coisol):ms}}	& \qmark(\isol(a , A , w) \sync \coisol(a , A , w))
				\\	\xproceq{SMA,$\isol[]$,$\coisol[]$}									& \qmark(a \mact \substfun(a) \mact \cosubstfun(a))
				\\	\xeq{$\qmark[]$}													& \block[\substenvimg(\substenv)](\hide[\set{\silentact}](\comm[\substenvcomm(\substenv)](a \mact \substfun(a) \mact \cosubstfun(a))))
				\\	\xproceq{C1,SMA}													& \block[\substenvimg(\substenv)](\hide[\set{\silentact}](a \sync \silentact))
				\\	\xproceq{H4,H3,H2}													& \block[\substenvimg(\substenv)](a \sync \silent) \eqeqbox{\xproceq{B4,SMA}} a \mact \silent \eqeqbox{\xproceq{MA3}} a \eqeqbox{\xproceq{$\A$,$\alpha$}} \alpha}
				
				\item[Step:]
				$\alpha = \alpha_1 \mact \alpha_2$. Suppose that this proposition holds for $\alpha_1$ (IH1) and $\alpha_2$ (IH2).
				\eqproof																{ \qmark(\bisect(\alpha , A , w))
				\\	\xeq{$\alpha$,$\bisect[]$}											& \qmark(\isol(\alpha_1 \mact \alpha_2 , A , w) \parr \coisol(\alpha_1 \mact \alpha_2 , A , w))
				\\	\xproceq{Prem$\rightarrow$Lem.~\ref{lemma:qmark(isol,coisol):ms}}	& \qmark(\isol(\alpha_1 \mact \alpha_2 , A , w) \sync \coisol(\alpha_1 \mact \alpha_2 , A , w))
				\\	\xproceq{SMA,$\isol[]$,$\coisol[]$}									& \qmark(\isol(\alpha_1 , A , w) \mact \isol(\alpha_2 , A , w) \mact \coisol(\alpha_1 , A , w) \mact \coisol(\alpha_2 , A , w))
				\\	\xeq{$\qmark[]$}													& \block[\substenvimg(\substenv)](\hide[\set{\silentact}](\comm[\substenvcomm(\substenv)](
				\\																		& \quad \isol(\alpha_1 , A , w) \mact \isol(\alpha_2 , A , w) \mact \coisol(\alpha_1 , A , w) \mact \coisol(\alpha_2 , A , w))))
				\\	\xproceq{Cor.~\ref{corol:comm(mact):addit}}							& \block[\substenvimg(\substenv)](\hide[\set{\silentact}](
				\\																		& \quad \comm[\substenvcomm(\substenv)](\isol(\alpha_1 , A , w) \mact \coisol(\alpha_1 , A , w)) \mact {}
				\\																		& \quad \comm[\substenvcomm(\substenv)](\isol(\alpha_2 , A , w) \mact \coisol(\alpha_2 , A , w))))
				\\	\xproceq{SMA,B4,H4}													& \block[\substenvimg(\substenv)](\hide[\set{\silentact}](\comm[\substenvcomm(\substenv)](\isol(\alpha_1 , A , w) \mact \coisol(\alpha_1 , A , w)))) \sync {}
				\\																		& \block[\substenvimg(\substenv)](\hide[\set{\silentact}](\comm[\substenvcomm(\substenv)](\isol(\alpha_2 , A , w) \mact \coisol(\alpha_2 , A , w))))
				\\	\xproceq{$\qmark[]$,SMA}											& \qmark(\isol(\alpha_1 , A , w) \sync \coisol(\alpha_1 , A , w)) \sync \qmark(\isol(\alpha_2 , A , w) \sync \coisol(\alpha_2 , A , w))
				\\	\xproceq{Prem$\rightarrow$Lem.~\ref{lemma:qmark(isol,coisol):ms}}	& \qmark(\isol(\alpha_1 , A , w) \parr \coisol(\alpha_1 , A , w)) \sync \qmark(\isol(\alpha_2 , A , w) \parr \coisol(\alpha_2 , A , w))
				\\	\xproceq{$\bisect[]$}												& \qmark(\bisect(\alpha_1 , A , w)) \sync \qmark(\bisect(\alpha_2 , A , w)) \eqeqbox{\xproceq{IH1,IH2}} \alpha_1 \sync \alpha_2 \eqeqbox{\xproceq{SMA}} \alpha_1 \mact \alpha_2 \eqeqbox{\xeq{$\alpha$}} \alpha}
			\end{description}
		\end{minipage}
	}
	
	\caption{Proof of Theorem~\ref{theorem:bisect(mact):id}}
	\label{fig:proof:theorem:bisect(mact):id}
\end{figure}

\subsubsection{Theorems for processes}
\label{sect:split:thm:seq}

The following theorem generalizes Theorem~\ref{theorem:bisect(mact):id} from multi-actions to processes in $\seqprocuniv$: it states that splitting such a syntactically $\silent$-free process equals that process.
\begin{theorem}
	\label{theorem:bisect(seqproc):id}
	For all $\substenv = (\dot P \mapsto \dot p , \mathbb{A} , \silentact , \substfun[] , \cosubstfun[])$ such that $\acts(p) \cup A \subseteq \mathbb{A}$,
	\begin{center}
		$\big[\silentfree(p) \AND p \in \seqprocuniv \big] \IMPLIES \bisect(p , A , w) \proceq p$
	\end{center}
\end{theorem}

As for Theorem~\ref{theorem:bisect(mact):id}, we need some auxiliary lemmas to prove this theorem.  We formulate these lemmas below; proofs, as well as additional propositions on which these proofs rely, appear in \ifthenelse{\boolean{isTechReport}}{Section~\ref{sect:proofs}}{\cite{JCP12}}. The first lemma states the additivity property $\qmark(r_1 \sync r_2) = \qmark(r_1) \sync \qmark(r_2)$ when $r_1$ and $r_2$ denote the isolation and the coisolation of the processes $p$ and $q$. Importantly, while $p$ and $q$ may denote the same process, their isolation and coisolation must involve different strings over $\set{\one , \two}$ for the additivity to hold.
\begin{lemma}
	\label{lemma:qmark(isol(seqproc)|coisol(seqproc)):addit}
	For all $\substenv = (\dot P \mapsto \dot p , \mathbb{A} , \silentact , \substfun[] , \cosubstfun[])$ such that $\acts(p) \cup \acts(q) \cup A \subseteq \mathbb{A}$,
	\begin{center}
		\begin{tabular}[t]{@{}l@{}}
			$\big[ \silentfree(p) \AND \silentfree(q) \AND p , q \in \seqprocuniv \AND w \neq v \big] \IMPLIES$
		\\	$\quad \qmark(\isol(p , A , w) \sync \coisol(q , A , v)) \proceq \qmark(\isol(p , A , w)) \sync \qmark(\coisol(q , A , v))$
		\end{tabular}
	\end{center}
\end{lemma}
\begin{proof}
	[Proof (sketch)]
	The actions occurring in $\isol(p , A , w)$ differ from those occurring in $\isol(q , A , v)$ (except for the original actions in $p$ and $q$) because $\substfun[]$ and $\cosubstfun[]$ have disjoint images by their definition \emph{and} because $w \neq v$. In that case, there exists no communication in $\isol(p , A , w) \sync \coisol(p , A , w)$ that did not exist already in $\isol(p , A , w)$ or in $\coisol(p , A , w)$, enabling one to distribute $\comm[]$ in $\qmark[]$ among them. We can do the same for $\hide[]$ and $\block[]$ in $\qmark[]$ (by B4 and H4). See \ifthenelse{\boolean{isTechReport}}{Section~\ref{sect:proofs:lemma:qmark(isol(seqproc)|coisol(seqproc)):addit}}{\cite{JCP12}} for a detailed proof.
\end{proof}
\noindent The following lemma states that the process $\qmark(r)$ deadlocks when $r$ denotes only the isolation or only the coisolation of a process $p$.
\begin{lemma}
	\label{lemma:qmark(isol(seqproc)):dead}
	For all $\substenv = (\dot P \mapsto \dot p , \mathbb{A} , \silentact , \substfun[] , \cosubstfun[])$ such that $\acts(p) \cup A \subseteq \mathbb{A}$,
	\begin{center}
		$\big[ \silentfree(p) \AND p \in \seqprocuniv \big] \IMPLIES \big[ \qmark(\isol(p , A , w)) \proceq \dead \AND \qmark(\coisol(p , A , w)) \proceq \dead \big]$
	\end{center}
\end{lemma}
\begin{proof}
	[Proof (sketch)]
	Similar to the proof sketch of Lemma~\ref{lemma:qmark(isol,coisol):ms}. See \ifthenelse{\boolean{isTechReport}}{Section~\ref{sect:proofs:lemma:qmark(isol(seqproc)):dead}}{\cite{JCP12}} for a detailed proof.
\end{proof}

Suppose that we have two sequential processes, namely $r_1 = \isol(p , A , w) \seq \isol(q , A , w)$ and $r_2 = \coisol(p , A , w) \seq \coisol(q , A , w)$. Moreover, suppose that we take their parallel composition $r_1 \parr r_2$. Our final lemma states that instead of taking this parallel composition, one can compose the parallel composition $r^\sharp = \isol(p , A , w) \sync \coisol(p , A , w)$ and the parallel composition $r^\flat = \isol(q , A , w) \sync \coisol(q , A , w)$ sequentially and get the same process. In short: $r^\sharp \seq r^\flat$ equals $r_1 \parr r_2$.
\begin{lemma}
	\label{lemma:qmark(seq,sync,seq):argswap}
	For all $\substenv = (\dot P \mapsto \dot p , \mathbb{A} , \silentact , \substfun[] , \cosubstfun[])$ such that $\acts(p) \cup A \subseteq \mathbb{A}$,
	\begin{center}
		\begin{tabular}[t]{@{}l@{}}
			$\big[\silentfree(p) \AND p \in \seqprocuniv \big] \IMPLIES$
		\\	$\quad \qmark((\isol(p , A , w) \seq \isol(q , A , w)) \parr {} (\coisol(p , A , w) \seq \coisol(q , A , w))) \proceq {}$
		\\	$\quad \quad \qmark(\isol(p , A , w) \parr \coisol(p , A , w) \seq \isol(q , A , w) \parr \coisol(q , A , w))$
		\end{tabular}
	\end{center}
\end{lemma}
\begin{proof}
	[Proof (sketch)]
	The processes $\isol(p , A , w)$ and $\coisol(p , A , w)$ always stay ``synchronized'' when composed in parallel due to the $\qmark[]$ operator. This implies that these processes ``finish'' at the same time. Consequently, $\isol(q , A , w)$ and $\coisol(q , A , w)$ start at the same time, which implies the desired result. See \ifthenelse{\boolean{isTechReport}}{Section~\ref{sect:proofs:lemma:qmark(seq,sync,seq):argswap}}{\cite{JCP12}} for a detailed proof.
\end{proof}
\noindent Finally, Figure~\ref{fig:proof:theorem:bisect(seqproc):id} shows a proof of Theorem~\ref{theorem:bisect(seqproc):id}. Our last theorem generalizes Theorem~\ref{theorem:bisect(seqproc):id} from sequential processes to parallel processes; Figure~\ref{fig:proof:theorem:bisect(proc):id} shows a proof.

\begin{figure}[t]
	\centering
	\small
	\fbox{
		\begin{minipage}{.95\textwidth}
			Suppose $\big[ \silentfree(p) \AND p \in \seqprocuniv \big]$ (Prem). We have:
			\begin{center}
				$\bisect(p , A , w) \eqeqbox{\xeq{$\bisect[]$}} \qmark(\isol(p , A , w) \parr \coisol(p , A , w)) \eqeqbox{\xproceq{Prem$\rightarrow$Lem.~\ref{lemma:qmark(isol,coisol):ms}}} \qmark(\isol(p , A , w) \sync \coisol(p , A , w))$
			\end{center}
			Denote this property by Obs. We proceed by structural induction on $p$.
			\begin{description}
				\item[Base:]
				$p = \ALPHA$. If $\ALPHA = \alpha$, the theorem follows by Theorem~\ref{theorem:bisect(mact):id}. If $\ALPHA = \dead$:
				\begin{center}
					$\bisect(p , A , w) \eqeqbox{\xeq{Obs,$p$}} \qmark(\isol(\dead , A , w) \sync \coisol(\dead , A , w)) \eqeqbox{\xeq{$\isol[]$,$\coisol[]$}} \qmark(\dead \sync \dead) \eqeqbox{\xproceq{S4,Q2}} \dead \eqeqbox{\xeq{$\ALPHA$,$p$}} p$
				\end{center}
				
				\item[Step:]
				$p = p_1 \binop p_2$ with $\binop \in \set{\ch , \seq}$. Suppose that this theorem holds for $p_1$ (IH1) and $p_2$ (IH2).
				\begin{description}
					\item[Case:]
					$p = p_1 \ch p_2$.
					\eqproof																					{ \bisect(p , A , w)
					\\	\xproceq{Obs,$p$}																		& \qmark(\isol(p_1 \ch p_2 , A , w) \sync \coisol(p_1 \ch p_2 , A , w))
					\\	\xeq{$\isol[]$,$\coisol[]$}																& \qmark((\isol(p_1 , A , w\one) \ch \isol(p_2 , A , w\two)) \sync (\coisol(p_1 , A , w\one) \ch \coisol(p_2 , A , w\two)))
					\\	\xproceq{S7}																			& \qmark(\isol(p_1 , A , w\one) \sync \coisol(p_1 , A , w\one) \ch \isol(p_1 , A , w\one) \sync \coisol(p_2 , A , w\two) \ch {}
					\\																							& \quad \isol(p_2 , A , w\two) \sync \coisol(p_1 , A , w\one) \ch \isol(p_2 , A , w\two) \sync \coisol(p_2 , A , w\two))
					\\	\xproceq{Q3}																			& \qmark(\isol(p_1 , A , w\one) \sync \coisol(p_1 , A , w\one)) \ch \qmark(\isol(p_1 , A , w\one) \sync \coisol(p_2 , A , w\two)) \ch {}
					\\																							& \qmark(\isol(p_2 , A , w\two) \sync \coisol(p_1 , A , w\one)) \ch \qmark(\isol(p_2 , A , w\two) \sync \coisol(p_2 , A , w\two))
					\\	\xproceq{Obs}																			& \qmark(\bisect(p_1 , A , w\one)) \ch \qmark(\isol(p_1 , A , w\one) \sync \coisol(p_2 , A , w\two)) \ch {}
					\\																							& \qmark(\isol(p_2 , A , w\two) \sync \coisol(p_1 , A , w\one)) \ch \bisect(p_2 , A , w\two)
					\\	\xproceq{Prem$\rightarrow$Lem.~\ref{lemma:qmark(isol(seqproc)|coisol(seqproc)):addit}}	& \bisect(p_1 , A , w\one) \ch \qmark(\isol(p_1 , A , w\one)) \sync \qmark(\coisol(p_2 , A , w\two)) \ch {}
					\\																							& \qmark(\isol(p_2 , A , w\two)) \sync \qmark(\coisol(p_1 , A , w\one)) \ch \bisect(p_2 , A , w\two)
					\\	\xproceq{Prem$\rightarrow$Lem.~\ref{lemma:qmark(isol(seqproc)):dead}}					& \bisect(p_1 , A , w\one) \ch \dead \sync \dead \ch \dead \sync \dead \ch \bisect(p_2 , A , w\two) \eqeqbox{\xproceq{IH1,IH2,S4,A6}} p_1 \ch p_2 \eqeqbox{\xproceq{$p$}} p}
					
					\item[Case:]
					$p = p_1 \seq p_2$.
					\eqproof																	{ \bisect(p , A , w)
					\\	\xeq{$p$,$\bisect[]$}													& \qmark(\isol(p_1 \seq p_2 , A , w) \parr \coisol(p_1 \seq p_2 , A , w))
					\\	\xeq{$\isol[]$,$\coisol[]$}												& \qmark((\isol(p_1 , A , w) \seq \isol(p_2 , A , w)) \parr (\coisol(p_1 , A , w) \seq \coisol(p_2 , A , w)))
					\\	\xproceq{Prem$\rightarrow$Lem.~\ref{lemma:qmark(seq,sync,seq):argswap}}	& \qmark(\isol(p_1 , A , w) \parr \coisol(p_1 , A , w) \seq \isol(p_2 , A , w) \parr \coisol(p_2 , A , w))
					\\	\xproceq{Q4}															& \qmark(\isol(p_1 , A , w) \parr \coisol(p_1 , A , w)) \seq \qmark(\isol(p_2 , A , w) \parr \coisol(p_2 , A , w))
					\\	\xproceq{$\bisect[]$}													& \bisect(p_1 , A , w) \seq \bisect(p_2 , A , w) \eqeqbox{\xproceq{IH1,IH2}} p_1 \seq p_2 \eqeqbox{\xproceq{$p$}} p}
				\end{description}
			\end{description}
		\end{minipage}
	}
	
	\caption{Proof of Theorem~\ref{theorem:bisect(seqproc):id}.}
	\label{fig:proof:theorem:bisect(seqproc):id}
\end{figure}

\begin{theorem}
	\label{theorem:bisect(proc):id}
	For all $\substenv = (\dot P \mapsto \dot p , \mathbb{A} , \silentact , \substfun[] , \cosubstfun[])$ such that $\acts(p) \cup A \subseteq \mathbb{A}$,
	\begin{center}
		$\silentfree(\dot p) \IMPLIES \bisect(\dot p , A , w) \proceq \dot p$
	\end{center}
\end{theorem}

\begin{figure}[t]
	\newcommand{\WIDTH}{.75cm}
	\centering
	\small
	\fbox{
		\begin{minipage}{.95\textwidth}
			First, we prove $\bisect(p , A , w) \proceq p[\refbisect(\dot P , A , w) / \dot P]$ for all $p$ such that $\silentfree(p)$ and in which only $\dot P$ occurs as a reference (Prem), where $p[Q/R]$ denotes the syntactic substitution of the process reference $Q$ for the process reference $R$ in $p$, by structural induction on $p$.
			\begin{description}
				\item[Base:]
				$p \in \seqprocuniv$ or $p = \dot P$. If $p \in \seqprocuniv$, this theorem follows by Theorem~\ref{theorem:bisect(seqproc):id}. If $p = \dot P$:
				\eqprooff									{ \bisect(p , A , w)}
					{\xeq{$p$}}								{ \bisect(\dot P , A  ,w) \eqeqbox{\xeq{$\bisect[]$}} \refbisect(\dot P , A , w)
				\\	\xeq{$[Q/R]$}							& \dot P[\refbisect(\dot P , A , w) / \dot P] \eqeqbox{\xeq{$p$}} p[\refbisect(\dot P , A , w) / \dot P]}

				\item[Step:]
				$p = p_1 \binop p_2$ or $p = \fun(p_1)$ for $\binop \in \set{\seq , \ch , \parr , \lmerge , \sync}$ and $\fun \in \set{\restr , \block , \rename , \comm , \hide}$. Suppose that this lemma holds for $p_1$ (IH1) and $p_2$ (IH2). We proceed by case distinction.
				\begin{description}
					\item[Case:]
					$p = p_1 \binop p_2$. 
					\eqprooff				{ \bisect(p)}
						{\xeq{$p$}}			{ \bisect(p_1 \binop p_2) \eqeqbox{\xeq{$\bisect[]$}} \bisect(p_1) \binop \bisect(p_2)
					\\	\xproceq{IH1,IH2}	& p_1[\refbisect(\dot P , A , w) / \dot P] \binop p_2[\refbisect(\dot P , A , w) / \dot P]
					\\	\xeq{$[Q/R]$}		& (p_1 \binop p_2)[\refbisect(\dot P , A , w) / \dot P] \eqeqbox{\xeq{$p$}} p[\refbisect(\dot P , A , w) / \dot P]}

					\item[Case:]
					$p = \fun(p_1)$.
					\eqprooff				{ \bisect(p)}
						{\xeq{$p$}}			{ \bisect(\fun(p_1)) \eqeqbox{\xeq{$\bisect[]$}} \fun(\bisect(p_1))
					\\	\xproceq{IH1}		& \fun(p_1[\refbisect(\dot P , A , w) / \dot P]) 
					\\	\xeq{$[Q/R]$}		& \fun(p_1)[\refbisect(\dot P , A , w) / \dot P] \eqeqbox{\xeq{$p$}} p[\refbisect(\dot P , A , w) / \dot P]}
				\end{description}
			\end{description}
			Recall $\dot P \mapsto \dot p$ (such that only $\dot P$ occurs as a process reference in $\dot p$---see Section~\ref{sect:mcrl2}) and $\refbisect(\dot P , A , w) \mapsto \bisect(\dot p , A , w)$. To establish the equality of $\dot p$ and $\bisect(\dot p , A , w)$, i.e., $\dot p[\refbisect(\dot P , A , w) / \dot P]$, we must show that there exists a process operator $\Phi$ of which $\dot P$ and $\refbisect(\dot P , A , w)$ are fixed points (see also Section~9.6 in~\cite{GM10}). Let $\Phi = \lambda Z \bullet \dot p[Z / \dot P]$. It follows that $\dot P = \Phi \, \dot P$ and that $\refbisect(\dot P , A , w) = \Phi \, \refbisect(\dot P , A , w)$.
		\end{minipage}
	}
	
	\caption{Proof of Theorem~\ref{theorem:bisect(proc):id}}
	\label{fig:proof:theorem:bisect(proc):id}
\end{figure}

%
%%
%%%
\section{Application: Splitting Connectors}
\label{sect:appl}
\subsection{Formalization of (A)synchronous Regions}

We provide a formal definition of the synchronous regions of a connector, based on the mCRL2 semantics of Reo.  Let $p$ denote a process describing the behavior of a Reo connector, and let ${\longrightarrow}$ denote its transition relation (labeled with multi-actions). Recall that every action in $p$ represents a channel end or a node end. Let $a \in \acts(p)$ denote one such an end. We define the $a$-synchronous region of $p$ as the smallest set $X_a \subseteq \acts(p)$ such that:
\begin{itemize}
	\item $a \in X_a$.
	\item If $b \in X_a$ then $\acts(\beta) \subseteq X_a$ for all $\beta$ such that $q \xrightarrow{\beta} q'$ and $b \in \acts(\beta)$.
	\item If $b \in X_a$ then $\acts(\beta') \subseteq X_a$ for all $\beta,\beta'$ such that $q \xrightarrow{\beta} q'$ and $q \xrightarrow{\beta'} q''$ and $b \in \acts(\beta)$.
\end{itemize}
The second rule states that all the ends that occur in the same multi-action belong to the same synchronous region. The third rule states that all the ends that can have flow in some state $q$, but possibly in different transitions leaving $q$, belong to the same synchronous region. In that case, channel ends may exclude each other from flow, which requires them to synchronize and communicate about their behavior.

To exemplify the previous definition, consider the connector modeled by the process $p = a \mact b \seq c \ch d$. Informally, either this connector has flow through $a$ and $b$, followed by flow through $c$, or it has flow through $d$. We construct its $a$-synchronous region starting from the singleton set $X_a = \set{a}$ (first rule). Subsequently, due to the multi-action $a \mact b$, we add $b$ to this set (second rule). The transition system of $p$ contains a state with two outgoing transitions: one labeled by $a \mact b$, the other labeled by $d$. Hence, because $a \in X_a$, we add $d$ to $X_a$ (third rule). This concludes the construction: $X_a = X_b = X_d = \set{a , b , d}$.

We define the set of the synchronous regions of the connector modeled by a process $p$ as
\begin{center}
	$\mathcal{X} = \bigcup_{a \in \acts(p)} \set{X_a}$
\end{center}
and the set containing its asynchronous regions as
\begin{center}
	$\mathcal{Y} = \setbuild{(a,b)}{\mathit{connected}(a,b) \mbox{ and } a \in X \mbox{ and } b \in X' \mbox{ and } X \neq X' \mbox{ and } X,X' \in \mathcal{X}}$,
\end{center}
where $\mathit{connected}(a,b)$ denotes that the ends $a$ and $b$ belong to the same channel.

\subsection{Splitting Connectors}
\label{sect:appl:split}

We set out to establish the soundness of splitting connectors along the boundaries of their (a)synchronous regions. However, we can split any (syntactically $\silent$-free) process along any set of actions by Theorem~\ref{theorem:bisect(proc):id}. This suggests that regardless of its (a)synchronous regions, one can split a connector in any possible way \emph{and} preserve its original semantics. While true in theory, there is a catch for implementations of splitted connectors in practice: the parallel composition of the isolation and the coisolation of a connector process must synchronize, represented by the $\qmark[]$ operator in Definition~\ref{def:bisect}. Depending on the particular implementation approach, which in turn may depend on the underlying hardware architecture (see Section~\ref{sect:intr}), performing $\qmark[]$ at run-time may cost an unreasonable amount of resources, if possible at all. Next, we demonstrate that arbitrary splitting, therefore, makes no sense in practice despite its theoretical validity. Splitting based on (a)synchronous regions, in contrast, does.

We start with an example of splitting based on (a)synchronous regions. Suppose that we split \Fifo{a,b} into two parts: one part contains only $a$, while the other part contains only $b$. Recall from Section~\ref{sect:reo} that the semantics of this channel is given by the process definition $\procfifo{a;b} \mapsto a \cdot b \cdot \procfifo{a;b}$. Splitting along $\set{a}$ (or equivalently, along $\set{b}$) yields:
\begin{center}
	$\begin{array}[t]{@{} l @{\;} c @{\;} r @{\,} r @{\,} r @{\,} l @{}}
		\refbisect(\procfifo{a;b} , \set{a} , \emptystring)	& \mapsto	& \multicolumn{4}{@{}l@{}}{\bisect(a \seq b \seq \procfifo{a;b} , \set{a} , \emptystring)}
	\\														& =			& \multicolumn{4}{@{}l@{}}{\bisect(a \seq b , \set{a} , \emptystring) \seq \bisect(\procfifo{a;b} , \set{a} , \emptystring)}
	\\														& =			& \qmark(a \mact \substfun[\emptystring](a) \seq {}	& \cosubstfun[\emptystring](b)			& \parr	& {}
	\\														&			& \cosubstfun[\emptystring](a) \seq {}				& b \mact \substfun[\emptystring](b)	&		& ) \seq \refbisect(\procfifo{a;b} , \set{a} , \emptystring))
	\end{array}$
\end{center}
with $\substenv = (\procfifo{a;b} \mapsto a \seq b \seq \procfifo{a;b} , \set{a , b} , \silentact , \substfun[] , \cosubstfun[])$. Here, $\qmark[]$ represents the asynchronous region of \Fifo{a;b}. Suppose that we want to implement $p = a \mact \substfun[\emptystring](a) \seq \cosubstfun[\emptystring](b)$ and $q = \cosubstfun[\emptystring](a) \seq b \mact \substfun[\emptystring](b)$ such that, when run in parallel, they behave as $a \seq b$. These implementations should perform the synchronization implied by $\qmark[]$. Recall from Section~\ref{sect:split} that intuitively, $\substfun[\emptystring](a)$ represents the act of ``disseminating the performance of $a$,'' while $\cosubstfun[\emptystring](a)$ represents the act of ``discovering the performance of $a$.'' Thus, the implementation of $p$ should: (1) accept data on $a$ and disseminate this acceptance, and (2) discover the dispersal of data on $b$. Meanwhile, the implementation of $q$ should: (1) discover the acceptance of data on $a$, and (2) dispense data on $b$ and disseminate this dispersal. Thus, in each step, the implementations of $p$ and $q$ require only unidirectional communication about their behavior to synchronize: first, the implementation of $p$ performs $\substfun[\emptystring](a)$ and the implementation of $q$ takes notice of this (by performing $\cosubstfun[\emptystring](a)$); afterwards, $p$ and $q$ switch roles to perform $\cosubstfun[\emptystring](b)$ and $\substfun[\emptystring](b)$. This shows that synchronous regions can decide on their behavior independently of each other: the region $\set{a}$ does not need to know that the region $\set{b}$ dispenses data before it can accept data---it can decide to do so without communication. 

In practice, this can yield performance improvements: although the isolation and the coisolation of a process $p$ have the same transition system modulo transition labels, benefits can arise if one composes them in parallel with \emph{another} split process $q$. In that case, there may exist a transition $t$ of the (co)isolation of $p$ that can proceed independently---without communication among the ends involved---of a transition $t'$ of the (co)isolation of $q$. Without splitting, in contrast, communication among those ends must always take place to decide on whether to behave according to $t$, $t'$, or both. For instance, if we put two split \Fifo{} instances in sequence (as in Figure~\ref{fig:conn:fifo2}), the source end $a$ of the first \Fifo{} can proceed independently of the sink end $b$ of the second \Fifo{}. This means that, if empty, the first \Fifo{} can accept a data item on $a$ (and place it in its buffer) without communicating with $b$. Similarly, if full, the second \Fifo{} can dispense a data item on $b$ (and remove it from its buffer) without communicating with $a$. In contrast, if we put two unsplit \Fifo{} instances in sequence, the source end $a$ and the sink end $b$ communicate with each other to decide on their joint behavior, even though the behavior of those ends does not depend on each other. By splitting, one avoids this unnecessary communication, reducing resource consumption at runtime.

To demonstrate that splitting arbitrarily makes no sense, suppose that we split \Sync{a,b} into two parts: one part contains only $a$, while the other part contains only $b$. Recall from Section~\ref{sect:reo} that the semantics of this channel is given by the process definition $\procsync{a;b} \mapsto a \mact b \seq \procsync{a;b}$. Splitting along $\set{a}$ (or equivalently, along $\set{b}$) yields:
\begin{center}
	$\begin{array}[t]{@{} l @{\;} c @{\;} r @{\,} r @{\,} l @{\,} l @{}}
		\refbisect(\procsync{a;b} , \set{a} , \emptystring)	& \mapsto	& \multicolumn{4}{@{}l@{}}{\bisect(a \mact b \seq \procsync{a;b} , \set{a} , \emptystring)}
	\\														& =			& \multicolumn{4}{@{}l@{}}{\bisect(a \mact b , \set{a} , \emptystring) \seq \bisect(\procsync{a;b} , \set{a} , \emptystring)}
	\\														& =	& \qmark(a \mact \substfun[\emptystring](a) \mact {}	& \cosubstfun[\emptystring](b)			& \parr {}
	\\														&	& \cosubstfun[\emptystring](a) \mact {}					& b \mact \substfun[\emptystring](b)	&			& ) \seq \refbisect(\procsync{a;b} , \set{a} , \emptystring)
	\end{array}$
\end{center}
with $\substenv = (\procsync{a;b} \mapsto a \mact b \seq \procsync{a;b} , \set{a , b} , \silentact , \substfun[] , \cosubstfun[])$. Now, similar to the previous example, suppose that we want to implement $p = a \mact \substfun[\emptystring](a) \mact \cosubstfun[\emptystring](b)$ and $q = \cosubstfun[\emptystring](a) \mact b \mact \substfun[\emptystring](b)$ such that, when run in parallel, they behave as $a \mact b$. As before, these implementations should perform the synchronization implied by $\qmark[]$. Thus, the implementation of $p$ should accept data on $a$, disseminate this acceptance, and discover the dispersal of data on $b$. Meanwhile, the implementation of $q$ should discover the acceptance of data on $a$, dispense data on $b$, and disseminate this dispersal. All of these actions must occur at the same time. This means that, in contrast to our previous example, the implementations of $p$ and $q$ must engage in \emph{bi}directional communication with each other about the acceptance of data on $a$ and the dispersal of data on $b$. This suggests that the two ends of \Sync{a,b} must synchronize with each other---they belong to the same synchronous region and cannot decide on their behavior independently---making it unreasonable to split them in the first place: the communication necessary to realize the synchronization necessary inflicts overhead, making it more attractive to run the original \Sync{a,b} without splitting.

Depending on the hardware architecture, one can implement unidirectional communication efficiently; we sketch an implementation of the split \Fifo{a,b} on a shared memory machine with multi-threading. First, we instantiate two threads, $A$ and $B$, for the processes $p = a \mact \substfun[\emptystring](a) \seq \cosubstfun[\emptystring](b)$ and $q = \cosubstfun[\emptystring](a) \seq b \mact \substfun[\emptystring](b)$. Every multi-action $\alpha$ translates to the atomic execution of a block of code representing the actions occurring in $\alpha$. We implement the action $\substfun[\emptystring](a)$ as setting a shared Boolean flag and the action $\cosubstfun[\emptystring](a)$ as waiting for the value of this flag to change. Once the latter happens, thread $B$ unsets the flag and knows that thread $A$ has accepted data from $a$. Subsequently, it can dispense the data on $b$ and set another shared flag for the actions $\substfun[\emptystring](b)$ and $\cosubstfun[\emptystring](b)$. In general, rather than simple Boolean flags, threads can share more complex data structures to keep track of which actions they have performed. 

Now, suppose that \Fifo{a,b} constitutes some arbitrarily large connector with a distributed implementatation across multiple machines in a network. In the standard distributed approach (see Section~\ref{sect:intr}), the implementation of \Fifo{a,b} has to share information with each of its neighbors in every step. We can reduce the amount of communication necessary for this sharing (and improve performance) by using the implementation of the split \Fifo{a,b} as described above (under the assumption that the machine on which we run this implementation features multi-threading and shared memory). The validity of doing this follows from Theorem~\ref{theorem:bisect(proc):id}: $\block(\comm(\cdots \parr \procfifo{a,b} \parr \cdots)) \proceq \block(\comm(\cdots \parr \refbisect(\procfifo{a,b}) \parr \cdots))$.

%
%%
%%%
\section{Future Work}
\label{sect:conc}

We identify three main directions for future work.
\begin{itemize}
	\item Implementing the splitting procedure to facilitate automatic splitting of processes, as well as a tool for the automatic detection of (a)synchronous regions of Reo connectors. Combined, they allow us to mechanically split connectors along their (a)synchronous regions. We can then integrate this in one of the code generation frameworks currently under development for Reo.
	
	\item Extending the splitting procedure to full mCRL2, including data and time. We see no fundamental difficulties along this path, although we expect the technical details and proofs to involve rather cumbersome derivations.
	
	\item Investigating other ways of splitting processes. The procedure we introduced in this paper splits processes in a synchronous manner, meaning that the action $\substfun(a)$ occurs at the same time as the action $a$ itself. We imagine at least two other ways of splitting processes. In one approach, $\substfun(a)$ occurs after $a$ but before the next action. Then, the process $q = a \seq b$ has $a \seq \substfun(a) \seq \cosubstfun(b)$ as its $\set{a}$-isolation (instead of $a \mact \substfun(a) \seq \cosubstfun(b)$. In another approach, $\substfun(a)$ occurs after $a$ but possibly concurrently with the next action. Then, $q$ has $a \seq (\substfun(a) \parr \cosubstfun(b))$ as its isolation. We spectulate that these splitting approaches are sound only under equivalences weaker than strong bisimulation.
	
	This line of research seems related to existing work on delay-insensitive circuits (e.g., \cite{Udd84}) and desynchronization (e.g., \cite{BC10,FJ96}), the derivation of an asynchronous system from a synchronous system: for the class of \emph{desynchronizable systems}, the original synchronous system and the newly constructed asynchronous system are semantically equivalent. If we use the splitting procedure presented in our paper to obtain such an original synchronous system, we may use---perhaps with modifications---results from desynchronization for our splitting purpose.
\end{itemize}

\bibliographystyle{eptcs}
\bibliography{generic}

\markboth{\hfill Appendix}{Appendix \hfill}
\appendix
 
%
%%
%%%
\section{Axiomatization}
\label{sect:axioms}

Every process has an associated transition system describing its semantics (see \cite{GMR+08} for the SOS rules). Let $\proceq$ denote equality of processes. Figure~\ref{fig:axioms} shows a sound and complete axiomatization---for strong bisimulation---of the operators shown in Figure~\ref{fig:syntax}. Figure~\ref{fig:axioms:moremact} axiomatizes two additional operators on multi-actions. Informally, the operator~$\mactminus$ subtracts the multi-action on its right-hand side from the multi-action on its left-hand side; the operator~$\mactin$ checks if the multi-action on its right-hand side contains the multi-action on its left-hand side. 

The axioms C1 and CL1 in Figure~\ref{fig:axioms:addit} refer to several auxiliary functions; Figure~\ref{fig:auxfuns} shows their definitions. The function~$\commfun{}$ applies the communications in a set $C$ to a multi-action. The function~$\mactalph$ maps a basic process~$p$ to its \emph{alphabet}, i.e., the multi-actions that occur in $p$. The function~$\mactpowset$ maps a set of multi-actions $V$ to those nonempty multi-actions contained in at least one multi-action in $V$. Finally, the function~$\commdom$ maps a set of communications $C$ to their domains.

\begin{figure}[t]
	\centering
	\fbox{
		$\begin{array}{@{} l @{\enspace} c @{\enspace} l @{}}
			\commfun(\alpha)	& =	& \left\{ \begin{array}{@{} l @{\quad} l @{}}
										\commfun[C_1](\commfun[C_2](\alpha))		& \mbox{if } C = C_1 \cup C_2 \mbox{ and } C_1 \neq \emptyset \mbox{ and } C_2 \neq \emptyset
									\\	b \mact \commfun(\alpha \setminus \beta)	& \mbox{if } C = \set{\beta \rightarrow b} \mbox{ and } \beta \sqsubseteq \alpha
									\\	\alpha										& \mbox{otherwise}
									\end{array} \right.
		\\						&	& \mbox{\scriptsize for $C \subseteq \mactuniv \times \actuniv$ a set of communications}
		\\	\vspace{-1.5ex}
		\\	\mactalph(p)		& =	& \left\{ \begin{array}{@{} l @{\quad} l @{}}
										\set{\alpha}					& \mbox{if } p = \alpha
									\\	\emptyset						& \mbox{if } p \in \set{\silent , \dead}
									\\	\mactalph(q) \cup \mactalph(r) 	& \mbox{if } p = q \binop r \mbox{ with } \binop \in \set{\ch , \seq}
									\end{array} \right.
		\\	\vspace{-1.5ex}
		\\	\mactpowset(V)		& =	& \set{\beta \mactin \alpha \,|\, \alpha \in V} \setminus \set{\silent}  
		\\						&	& \mbox{\scriptsize for $V \subseteq \mactuniv$ a set of multi-actions}
		\\	\vspace{-1.5ex}
		\\	\commdom(C)			& =	& \left\{ \begin{array}{@{} l @{\quad} l @{}}
										\commdom(C_1) \cup \commdom(C_2)	& \mbox{if } C = C_1 \cup C_2 \mbox{ and } C_1 \neq \emptyset \mbox{ and } C_2 \neq \emptyset
									\\	\set{\beta}							& \mbox{if } C = \set{\beta \rightarrow b}
									\end{array} \right.
		\\						&	& \mbox{\scriptsize for $C \subseteq \mactuniv \times \actuniv$ a set of communications}
		\end{array}$
	}
	
	\caption{Auxiliary functions.}
	\label{fig:auxfuns}
\end{figure}

\begin{figure}[p]
	\newcommand{\HEIGHT}{154pt}
	\newcommand{\WIDTH}{.435\linewidth}
	\newcommand{\INNERWIDTH}{.425\linewidth}
	
	\hfil
	\subfloat[Axioms for multi-actions and for the basic operators.]{\label{fig:axioms:mact}\label{fig:axioms:basic}
		\fbox{
			\vbox to \HEIGHT {%
				\vfil
				\hbox to \WIDTH {
					\hfil
					\begin{minipage}{\INNERWIDTH}
						\centering
						\begin{tabular}{@{} r @{$\quad$} l @{}}
							MA1	& $\alpha \mact \beta \proceq \beta \mact \alpha$
						\\	MA2	& $(\alpha \mact \beta) \mact \gamma \proceq \alpha \mact (\beta \mact \gamma)$
						\\	MA3	& $\alpha \mact \silent \proceq \alpha$
						\\
						\\	A1	& $p \ch q \proceq q + p$
						\\	A2	& $p \ch (q \ch r) \proceq (p \ch q) \ch r$
						\\	A3	& $p \ch p \proceq p$
						\\	A4	& $(p \ch q) \seq r \proceq p \seq r \ch q \seq r$
						\\	A5	& $(p \seq q) \seq r \proceq p \seq (q \seq r)$
						\\	A6	& $p \ch \dead \proceq p$
						\\	A7	& $\dead \seq p \proceq \dead$
						\end{tabular}
					\end{minipage}
					\hfil
				}
				\vfil
			}
		}
	}
	\hfil
	\subfloat[More axioms for multi-actions.]{\label{fig:axioms:moremact}
		\fbox{
			\vbox to \HEIGHT {%
				\vfil
				\hbox to \WIDTH {
					\hfil
					\begin{minipage}{\INNERWIDTH}
						\centering
						\begin{tabular}{@{} r @{$\quad$} l @{}}
							MD1		& $\silent \setminus \alpha \proceq \silent$
						\\	MD2		& $\alpha \setminus \silent \proceq \alpha$
						\\	MD3		& $\alpha \setminus (\beta \mact \gamma) \proceq (\alpha \setminus \beta) \setminus \gamma$
						\\	MD4		& $(a \mact \alpha) \setminus a \proceq \alpha$
						\\	MD5		& $(a \mact \alpha) \setminus b \proceq a \mact (\alpha \setminus b)$ if $a \not = b$
						\\
						\\	MS1		& $\silent \sqsubseteq \alpha \proceq \true$
						\\	MS2		& $\alpha \sqsubseteq \silent \proceq \false$ if $\alpha \not \proceq \silent$
						\\	MS3		& $a \mact \alpha \sqsubseteq a \mact \beta \proceq \alpha \sqsubseteq \beta$
						\\	MS4		& $a \mact \alpha \sqsubseteq b \mact \beta \proceq$
						\\			& $\quad a \mact (\alpha \setminus b) \sqsubseteq \beta$ if $a \not = b$
						\end{tabular}
					\end{minipage}
					\hfil
				}
				\vfil
			}
		}
	}
	\hfil
	
	\renewcommand{\HEIGHT}{224pt}
	
	\hfil
	\subfloat[Axioms for the parallel operators.]{\label{fig:axioms:par}
		\fbox{
			\vbox to \HEIGHT {%
				\vfil
				\hbox to \WIDTH {
					\hfil
					\begin{minipage}{\INNERWIDTH}
						\centering
						\begin{tabular}{@{} r @{$\quad$} l @{}}
							M	& $p \parr q \proceq p \lmerge q \ch q \lmerge p \ch p \sync q$
						\\
						\\	LM1	& $\ALPHA \lmerge p \proceq \ALPHA \seq p$
						\\	LM2	& $\dead \lmerge p \proceq \dead$
						\\	LM3	& $\alpha \seq p \lmerge q \proceq \alpha \seq (p \parr q)$
						\\	LM4	& $(p \ch q) \lmerge r \proceq p \lmerge r \ch q \lmerge r$
						\\
						\\	S1	& $p \sync q \proceq q \sync p$
						\\	S2	& $(p \sync q) \sync r \proceq p \sync (q \sync r)$
						\\	S3	& $p \sync \silent \proceq p$
						\\	S4	& $\ALPHA \sync \dead \proceq \dead$
						\\	S5	& $(\ALPHA \seq p) \sync \BETA \proceq \ALPHA \sync \BETA \seq p$
						\\	S6	& $(\ALPHA \seq p) \sync (\BETA \seq q) \proceq \ALPHA \sync \BETA \seq (p \parr q)$
						\\	S7	& $(p \ch q) \sync r \proceq p \sync r \ch q \sync r$
						\\
						\\	SMA	& $\alpha \sync \beta \proceq \alpha \mact \beta$
						\end{tabular}
					\end{minipage}
					\hfil
				}
				\vfil
			}
		}
	}
	\hfil
	\subfloat[Axioms for the additional operators.]{\label{fig:axioms:addit}
		\fbox{
			\vbox to \HEIGHT {%
				\vfil
				\hbox to \WIDTH {
					\hfil
					\begin{minipage}{\INNERWIDTH}
						\centering
						\begin{tabular}{@{} r @{$\quad$} l @{}}
							V1	& $\restr(\alpha) \proceq \alpha$ if $\alpha \in V \cup \set{\silent}$
						\\	V2	& $\restr(\alpha) \proceq \dead$ if $\alpha \notin V \cup \set{\silent}$
						\\
						\\	B1	& $\block(\silent) \proceq \silent$
						\\	B2	& $\block(a) \proceq a$ if $a \notin B$
						\\	B3	& $\block(a) \proceq \dead$ if $a \in B$
						\\	B4	& $\block(\alpha \sync \beta) \proceq \block(\alpha) \sync \block(\beta)$
						\\	
						\\	R1	& $\rename(\silent) \proceq \silent$
						\\	R2	& $\rename(a) \proceq b$ if $a \rightarrow b \in R$ for some $b$
						\\	R3	& $\rename(a) \proceq a$ if $a \rightarrow b \notin R$ for all $b$
						\\	R4	& $\rename(\alpha \sync \beta) \proceq \rename(\alpha) \sync \rename(\beta)$
						\\
						\\	C1	& $\comm(\alpha) \proceq \commfun(\alpha)$
						\\
						\\	CL1	& $\comm(p) \proceq p$ if $\mactpowset (\mactalph(p)) \cap dom(C) = \emptyset$
						\end{tabular}
					\end{minipage}
					\hfil
				}
				\vfil
			}
		}
	}
	\hfil
	
	\renewcommand{\HEIGHT}{98pt}
	
	\hfil
	\subfloat[More axioms for the additional operators.]{\label{fig:axioms:moreaddit}
		\fbox{
			\vbox to \HEIGHT {%
				\vfil
				\hbox to \WIDTH {
					\hfil
					\begin{minipage}{\INNERWIDTH}
						\centering
						\begin{tabular}{@{} r @{$\quad$} l @{}}
							\multicolumn{2}{@{}l@{}}{For all $\fun \in \set{\restr , \block , \rename , \comm}$,}
						\\	\vspace{-1ex}
						\\	V3, B5, R5, C2	& $\fun(\dead) \proceq \dead$
						\\	V4, B6, R6, C3	& $\fun(\alpha \ch \beta) \proceq \fun(\alpha) \ch \fun(\beta)$
						\\	V5, B7, R7, C4	& $\fun(\alpha \seq \beta) \proceq \fun(\alpha) \seq \fun(\beta)$
						\end{tabular}
					\end{minipage}
					\hfil
				}
				\vfil
			}
		}
	}
	\hfil
	\subfloat[Axioms for the abstraction operator.]{\label{fig:axioms:abstr}
		\fbox{
			\vbox to \HEIGHT {%
				\vfil
				\hbox to \WIDTH {
					\hfil
					\begin{minipage}{\INNERWIDTH}
						\centering
						\begin{tabular}{@{} r @{$\quad$} l @{}}
							H1	& $\hide(\silent) \proceq \silent$
						\\	H2	& $\hide(a) \proceq \silent$ if $a \in I$
						\\	H3	& $\hide(a) \proceq a$ if $a \notin I$
						\\	H4	& $\hide(\alpha \sync \beta) \proceq \hide(\alpha) \sync \hide(\beta)$
						\\	H5	& $\hide(\dead) \proceq \dead$
						\\	H6	& $\hide(\alpha \ch \beta) \proceq \hide(\alpha) \ch \hide(\beta)$
						\\	H7	& $\hide(\alpha \seq \beta) \proceq \hide(\alpha) \seq \hide(\beta)$
						\end{tabular}
					\end{minipage}
					\hfil
				}
				\vfil
			}
		}
	}
	\hfil
	
	\caption{Axioms.}
	\label{fig:axioms}
\end{figure}

\ifthenelse{\boolean{isTechReport}}{%
Below, we list a number of auxiliary properties of multi-actions used directly and indirectly in Section~\ref{sect:proofs}. Proofs of these properties appear in Section~\ref{sect:proofs:mact}.
\begin{proposition}
	\label{prop:mact:nf}
	$ $
	\begin{center}
		$\left[ \begin{array}{@{}l@{}}
			\mactsize(\alpha) = 1 \OR
		\\	\quad \big[ \alpha \proceq a \mact \hat \alpha \AND \mactsize(\hat \alpha) < \mactsize(\alpha) \big]
		\end{array} \right] \mbox{ with } \mactsize(\alpha) = \left\{ \begin{array}{@{} l @{\quad} l @{}}
			1											& \mbox{if } \alpha = \A
		\\	\mactsize(\alpha_1) + \mactsize(\alpha_2)	& \mbox{if } \alpha = \alpha_1 \mact \alpha_2
		\end{array} \right.$
	\end{center}
\end{proposition}
\begin{proposition}
	\label{prop:mactin:ms3:gen}
	$\gamma \mact \alpha \mactin \gamma \mact \beta \proceq \alpha \mactin \beta$
\end{proposition}
\begin{proposition}
	\label{prop:mactin:beta:decomp:a}
	$a \mact \hat \alpha \mactin \beta \IMPLIES \beta \proceq a \mact \check \beta$
\end{proposition}
\begin{proposition}
	\label{prop:mactin:beta:elim}
	$a \mactin \beta_1 \mact \beta_2 \IMPLIES \big[ a \mactin \beta_1 \OR a \mactin \beta_2 \big]$
\end{proposition}
\begin{proposition}
	\label{prop:mactminus:md4:gen}
	$(\gamma \mact \alpha) \mactminus \gamma \proceq \alpha$
\end{proposition}
\begin{proposition}
	\label{prop:mactin:beta:decomp}
	$\alpha \mactin \beta \IMPLIES \beta \proceq \alpha \mact \tilde \beta$
\end{proposition}
\begin{proposition}
	\label{prop:mactin:alpha:elim}
	$\alpha_1 \mact \alpha_2 \mactin \beta \IMPLIES \alpha_1 \mactin \beta$
\end{proposition}
}{}

\ifthenelse{\boolean{isTechReport}}{\input{sect/proofs}}{}

\end{document}